\documentclass[journal,11pt,onecolumn,draftclsnofoot,]{IEEEtran}
\ifCLASSINFOpdf
\else
\fi
\usepackage{mathtools}
\usepackage{blindtext}
\usepackage{enumitem}
\usepackage{bbm, dsfont}
\usepackage{color}
\usepackage{tcolorbox}
\usepackage{balance}
\usepackage{amsthm}
\usepackage[utf8]{inputenc}
\usepackage{booktabs,caption}
\usepackage[flushleft]{threeparttable}
\usepackage[english]{babel}
\usepackage{amssymb}
\usepackage[utf8]{inputenc} 
\usepackage[T1]{fontenc}
\usepackage{dblfloatfix}
\usepackage{url}
\usepackage{tikz}
\usetikzlibrary{arrows}
\usetikzlibrary{automata,positioning}

\usepackage{ifthen}
\usepackage[noadjust]{cite}
\usepackage{graphicx}
\usepackage{psfrag}

\usepackage{accents}
\newcommand{\dbtilde}[1]{\accentset{\approx}{#1}}
\DeclarePairedDelimiter\ceil{\lceil}{\rceil}

\def\nn{\nonumber}

\allowdisplaybreaks
\begin{document}
%
\title{The Feedback Capacity of Noisy Output is the STate (NOST) Channels
\thanks{E. Shemuel was supported by the Ministry of Science and Technology of Israel. This work was supported by the German Research Foundation (DFG) via the German-Israeli Project Cooperation [DIP] and by the ISF research grant 818/17. 
O. Sabag is partially supported by the ISEF postdoctoral fellowship. E. Shemuel and H. H. Permuter are with the School of Electrical and
Computer Engineering, Ben-Gurion University of the Negev, Beersheba
8410501, Israel (e-mail: els@post.bgu.ac.il; haimp@bgu.ac.il). O. Sabag is with the Department of Electrical Engineering, Caltech, (e-mail: oron@caltech.edu).}}
%
%
%
\author{Eli Shemuel, Oron Sabag, Haim H. Permuter}
\maketitle
\begin{abstract}
We consider finite-state channels (FSCs) where the channel state is stochastically dependent on the previous channel output. We refer to these as Noisy Output is the STate (NOST) channels. We derive the feedback capacity of NOST channels in two scenarios: with and without causal state information (CSI) available at the encoder.
If CSI is unavailable, the feedback capacity is $C_{\text{FB}}= \max_{P(x|y')} I(X;Y|Y')$, while if it is available at the encoder, the feedback capacity is $C_{\text{FB-CSI}}= \max_{P(u|y'),x(u,s')} I(U;Y|Y')$, where $U$ is an auxiliary RV with finite cardinality. In both formulas, the output process is a Markov process with stationary distribution.
The derived formulas generalize special known instances from the literature, such as where the state is i.i.d. and where it is a deterministic function of the output. $C_{\text{FB}}$ and $C_{\text{FB-CSI}}$ are also shown to be computable via convex optimization problem formulations. Finally, we present an example of an interesting NOST channel for which CSI
available at the encoder does not increase the feedback capacity.
\end{abstract}

%
\IEEEpeerreviewmaketitle

\theoremstyle{plain}
\newtheorem{thm}{Theorem}[section]
\newtheorem{lem}[thm]{Lemma}
\newtheorem{prop}[thm]{Proposition}
\newtheorem*{cor}{Corollary}
\newtheorem*{KL}{Klein’s Lemma}

\newtheoremstyle{colon}%
{}
{}
{\itshape}
{}
{\itshape}
{:}
{ }
{}
\theoremstyle{colon}
\newtheorem{question}{Question}
\newtheorem{claim}{Claim}
\newtheorem{guess}{Conjecture}
\newtheorem{fact}{Fact}
\newtheorem{assumption}{Assumption}
\newtheorem{theorem}{Theorem}
\newtheorem{lemma}{Lemma}
\newtheorem{ctheorem}{Corrected Theorem}
\newtheorem{corollary}{Corollary}
\newtheorem{proposition}{Proposition}
\newtheorem{example}{Example}

\theoremstyle{definition}
\theoremstyle{remark}
\newtheorem{conj}{Conjecture}[section]
\newtheorem{exmp}{Example}[section]

\newtheoremstyle{def}%
{}
{}
{}
{}
{\itshape}
{:}
{ }
{}
\theoremstyle{def}
\newtheorem{definition}{Definition}
\newtheorem{remark}{Remark}
\newtheorem{discussion}{Discussion}

\theoremstyle{remark}
\newtheorem*{rem}{Remark}
\newtheorem*{note}{Note}
\newtheorem{case}{Case}

\def\cS{{\mathcal S}}
\def\cX{{\mathcal X}}
\def\cU{{\mathcal U}}
\def\cV{{\mathcal V}}
\def\cP{{\mathcal P}}
\def\cQ{{\mathcal Q}}
\def\cY{{\mathcal Y}}
\def\cC{{\mathcal C}}

\section{Introduction}
\IEEEPARstart{T}{he} popular model of finite-state channels (FSCs) \cite{Blackwell58,Gallager68,GBAA,Ziv85,PfisterISI,RGray} has been motivated by channels or systems with memory, 
common in wireless communication \cite{Sadeghi,ZHANG,turin1990performance,FSCTransWirelessComm,FSCTransWirelessComm1,pimentel2004finite,zhong2008model}, 
molecular communication \cite{MolecFSCTransComm,MolecularSurvey}
and magnetic recordings \cite{Immink}. 
The memory of a channel or a system is encapsulated in a finite set of states in the FSC model. Although feedback cannot increase the capacity of 
memoryless channels \cite{shannon56,Kadota--Zakai--Ziv1971a}, it can generally increase the capacity of channels with memory. Nonetheless, in the general case, both the capacity and the feedback capacity of FSCs 
were characterized by multi-letter expressions that are non-computable, and they still have no simple closed-form formulas. That is, 
a general capacity formula for channels with memory was given as the limit of the $n$-fold mutual information sequence \cite{Dobrusin59,Pinsker60,Gallager68,Verdu94}, whereas the feedback capacity was commonly expressed by the limit of the $n$-fold directed information \cite{massey1990causality,Kramer98,Kim07feedback,TatikondaMitter_IT09,PermuterWeissmanGoldsmith09,Permuter06_trapdoor_submit,ShraderPermuter09CompoundIT,dabora2012capacity,shemuel2018finite,shemuel2020feedback}.

The explicit capacity of FSCs is known only 
in a few instances of channels with memory where feedback does not increase the capacity, such as the POST($\alpha$) channel \cite{POSTchannel} and channels with certain symmetric properties \cite{Alajaji95,Alajaji94,SongAlajaji_Burst,SymmetricFSM}.
%
Further, a single-letter expression was derived in \cite{ViswanathanCapMarkov} for the capacity of FSCs with state information known at the receiver and delayed feedback in the absence of inter-symbol interference (ISI), i.e., 
the channel state is input-independent. 
There are some additional special cases of FSCs where the feedback capacity is known explicitly.
One method to compute explicit feedback capacity expression is by formulating it as a dynamic programming (DP) optimization
problem, as was first introduced in Tatikonda's thesis \cite{Tatikonda00} and then in \cite{Yang05,YangKavcicTatikondaGaussian,chen2005capacity,Permuter06_trapdoor_submit,TatikondaMitter_IT09,shemuel2018finite}. This is beneficial in 
estimating the feedback capacity using efficient algorithms such as the value iteration algorithm \cite{Bertsekas00}, which, in turn, can help in generating a conjecture for the exact solution of the corresponding Bellman equation \cite{Arapos93_average_cose_survey}.
Thus, for a family of FSCs with ISI called unifilar FSCs, in which the new channel state is a time-invariant function of the previous
state, the current input and the current output, the feedback capacity 
can be computed via DP, as was formulated in \cite{Permuter06_trapdoor_submit}, and closed-form expressions or exact values for the feedback capacity of particular unifilar FSCs were derived in \cite{Permuter06_trapdoor_submit,wu2016capacity,Sabag_BEC,Ising_channel,Sabag_BIBO_IT,Sabag_UB_IT,PeledSabagBEC,AharoniSabag_RL}.
For a sub-family of unifilar FSCs where the channel state is a deterministic function of the channel output, a single-letter feedback capacity expression was derived and formulated as a DP optimization problem in \cite{chen2005capacity}. 
Another method to compute explicit capacity expressions is the $Q$-graph that was introduced and utilized in \cite{Sabag_UB_IT,Graph-Based,SabagTwoWay,AharoniSabag_RL}. 

We are motivated to derive the feedback capacity of FSCs with ISI and stochastic state evolution as a single-letter, computable expression. In this paper, we consider a generalization of the unifilar FSCs studied in \cite{chen2005capacity}, where the channel state is stochastically dependent on the previous channel output; we refer this as "Noisy Output is the STate (NOST) channels". We study two scenarios of NOST channels
subject to causal state information (CSI) availability at the encoder, as illustrated in Fig. \ref{fig:setting}.
\begin{figure}[t]
\begin{center}
\begin{psfrags}
    \psfragscanon
    \psfrag{E}[][][1]{$M$}
    \psfrag{S}[][][1]
    {\raisebox{0.7cm}{$S^{i-1}$\hspace{0.1cm}}}
    \psfrag{T}[][][1]
    {\raisebox{0.2cm}{I\hspace{-0.1cm}}}
    \psfrag{U}[][][1]
    {\raisebox{0cm}{II\hspace{0.35cm}}}
    \psfrag{A}[\hspace{2cm}][][1]{Encoder}
	 \psfrag{F}[\hspace{1cm}][][1]{$X_i$}
	 \psfrag{B}[][][1]{$Q(y_i|x_i,s_{i-1})Q(s_i|y_i)$}
	 \psfrag{G}[][][1]{$Y_i$}
	 \psfrag{C}[\hspace{2cm}][][1]{Decoder}
	 \psfrag{K}[][][1]{$\hat{M}$}
	 \psfrag{H}[\hspace{2cm}][][1]{$Y_i$}
	 \psfrag{D}[\hspace{2cm}][][1]{Delay}
	 \psfrag{J}[\vspace{2cm}\hspace{2cm}][][1]{$Y_{i-1}$}
	 \psfrag{L}[\hspace{2cm}][][1]{Finite-State Channel}
	 \psfrag{I}[][][1]{}
\includegraphics[scale=0.8]{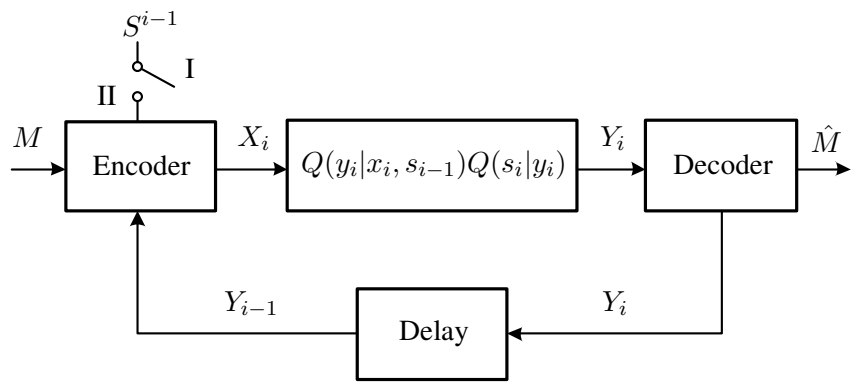}
\caption{
NOST channels in the presence of feedback. Setting I (open switch) - no CSI is available. Setting II (closed switch) - with CSI available at the encoder.}
\label{fig:setting}
\psfragscanoff
\end{psfrags}
\end{center}
\end{figure}
The first is where CSI is unavailable, while the second is where it is available at the encoder. For both these scenarios of NOST channels, the main contribution of this paper is single-letter, computable feedback capacity formulas. We show the computability of the formulas of both settings via formulating them as convex optimization problems. 
The achievability of the feedback capacity of the first setting is based on rate-splitting and random coding, and a similar proof was also given in \cite{ViswanathanCapMarkov}. A posterior matching scheme \cite{shayevitz_posterior_mathcing},
a principle that was also used in \cite{Sabag_BIBO_IT}, can also be used for the achievability from the work of \cite{bae_posterior_FSC}. On the other hand, the converse of the feedback formula is based on a recently developed technique to derive upper bounds with stationary distributions \cite{OronStationaryConverse}. In fact, the first setting can be shown to be equivalent to the setting mentioned in \cite{chen2005capacity}, and our formula is identical to theirs; however, our result generalizes and elaborates upon their seminal result in two ways. Firstly, the feedback capacity formula in \cite{chen2005capacity} is subject to an assumption that we relax in this work to a mild condition, thereby enabling us to determine the feedback capacity of various more channels, such as the POST($\alpha$) channel.
This condition plays a role in our derivation of the achievability.
The second contribution is the aforementioned computability of the formula via convex optimization.

The second setting in this paper, where CSI is available at the encoder, is innovative and interesting, as this side information may generally be beneficial for increasing the feedback capacity of channels with memory. 
The capacity problem of discrete memoryless channels (DMCs) with states known at the encoder dates back to Shannon's early work \cite{Shannon58}, followed by works of Kusnetsov and Tsybakov \cite{kuznetsov1974coding}, Gel’fand and Pinsker \cite{GePi80}, and Heegard and El Gamal \cite{HeegardElGamal_state_encoded83}, which paved the way to various recent works such as \cite{Weissman_action_dependtent_state10,choudhuri2013causal,WiretapChia,FeedbackCribbing,shemuel2018finite,shemuel2020feedback}.
Since Shannon showed in \cite{shannon56} that feedback does not increase the capacity of a DMC, his setting in \cite{Shannon58},
where the state process is i.i.d. and known causally at the encoder, is covered by our second setting by assuming, in particular, that the state is independent of the output. Further, we show that the capacity expression of this Shannon's setting is covered by ours. 

The remainder of the paper is organized as follows. Section \ref{sec:probDef} defines the notation and the settings. Section \ref{sec:main_results} presents the main results concerning the single-letter feedback capacity expressions and their convex optimization formulations. 
Section \ref{sec:specialCases} shows how the capacity expression of each setting covers the capacity characterization of known special cases from the literature,
and provides an interesting example of a connected NOST channel for which CSI available at the encoder does not increase its feedback capacity; this example is referred to as "the noisy-POST$(\alpha,\eta)$ channel", and is a generalization of the POST($\alpha$) channel \cite{POSTchannel}. 
Section \ref{sec:proof_cap} provides proofs and derivations of the main results. Finally, Section \ref{sec:conclusions} concludes this work.
%
%
%
%
%

\section{Problem Definition}
\label{sec:probDef}
In this section, we introduce the notation and the communication setup.
\subsection{Notation}
Lowercase letters denote sample values (e.g. $x,y$), and uppercase letters denote discrete random variables (RVs) (e.g. $X,Y$). Subscripts and superscripts denote vectors in the following way: $x_i^j=(x_i,x_{i+1},...,x_j)$ and $X_i^j=(X_i,X_{i+1},...,X_j)$ for $1\leq i \leq j$. $x^n$ and $X^n$ are shorthand for $x_1^n$ and $X_1^n$, respectively.
We use calligraphic letters (e.g. $\cX,\cY$) to denote alphabets, and $|\cdot|$ (e.g. $|\cX|$) to denote the cardinality of an alphabet. For two RVs $X,Y$ the probability mass function (PMF) of $X$ is denoted by $P_{X}(x)$, the conditional PMF of $X=x$ given $Y=y$ is denoted by $P_{X|Y}(x|y)$, and the joint PMF is denoted by $P_{X,Y}(x,y)$; the shorthand $P(x),P(x|y),P(x,y)$ are used for the above, respectively, when the RVs are clear from the context.
The indicator function is denoted by $\mathbbm{1}{\{\cdot\}}$. We define $\bar{a}\triangleq 1-a$ for some $a \in [0,1]$. 
For a pair of integers $n\le m$, we define the discrete interval $[n:m]\triangleq \{n,n+1,\dots,m\}$. 

\subsection{The Communication Setup}
\label{subsec:TheCommunicationSetup}
In this work, we consider FSCs as shown in Fig. \ref{fig:setting}.
An FSC \cite{Gallager68} consists of finite input, output and channel state alphabets $\cX,\cY,\cS$, respectively.
It is defined by the model ($\cX \times \cS$, $Q(y,s|x,s')$, $\cY$) where $s',s$ are the channel state at the beginning and at the end of a transmission, respectively.
The channel is stationary in the sense that when it is used $n$ times with message $M$ and inputs $X^n$, at time $i\in [1:n]$ given the past, it has the Markov property
\begin{align}
Q(y_i,s_i|x^i,s_0^{i-1},y^{i-1},m) 
&= Q_{Y,S|X,S'}(y_i,s_i|x_i,s_{i-1}) \label{eq:DefFSC}\\ 
&= Q_{Y|X,S'}(y_i|x_i,s_{i-1})Q_{S|Y}(s_i|y_i) \label{eq:DefNOST},
\end{align}
where \eqref{eq:DefFSC} holds for any general FSC and \eqref{eq:DefNOST} 
is particularized for NOST channels. We also use the averaged channel defined by 
\begin{align}\label{eq:y-xy}
  Q_{Y|X,Y'}(y|x,y') = \sum_{s'\in \cS} Q_{S|Y}(s'|y') Q_{Y|X,S'}(y|x,s'), 
\end{align}
where $y',y \in \cY$ can be interpreted as the channel output before and after a transmission, respectively.
The initial channel state is assumed to be distributed according to $Q(s_0)$ and will be shown to have no effect on the feedback capacity solution subject to a mild connectivity assumption. We consider two settings of NOST channels: without and with CSI available at the encoder. 
\subsubsection{Setting I - No CSI}
At time $i$, the encoder has access to the message $m\in \mathcal{M}$ and the outputs' feedback, where the message set is $\mathcal{M}=[1:\ceil{2^{nR}}]$ and $M$ is assumed to be uniformly distributed over $\mathcal{M}$. The encoder's mapping at time $i$ is denoted by 
\begin{align}
\label{eq:enc}
    x_i:\mathcal{M} \times \cY^{i-1} \to \cX,
\end{align}
and the decoder's mapping is
\begin{align}
\label{eq:dec}
    \hat{m}: \cY^n \to \mathcal{M}.
\end{align}
A $(2^{nR},n)$ code is a pair of encoding and decoding mappings \eqref{eq:enc}-\eqref{eq:dec}.
A rate $R$ is \textit{achievable} if there exists a sequence of $(2^{nR},n)$ codes such that the \textit{average probability of error}, $P_{e}^{(n)}=\Pr (\hat{M}\neq M)$, tends to zero as $n\to \infty$. The \textit{feedback capacity} is the supremum over all achievable rates, and is denoted by $C_{\text{FB}}$.

\subsubsection{Setting II - CSI Available at the Encoder}
Setting II is defined similarly to Setting I, 
except that at time $i$ the encoder also has access to all previous states, i.e., causally. Hence, the encoder's mapping at time $i$ is denoted by 
\begin{align}
\label{eq:enc2}
    x_i:\mathcal{M} \times \cS_{0}^{i-1} \times \cY^{i-1} \to \cX.
\end{align}
A $(2^{nR},n)$ code is a pair of encoding and decoding mappings given by \eqref{eq:enc2} and \eqref{eq:dec}, respectively. The feedback capacity of Setting II is denoted by $C_{\text{FB-CSI}}$. 

For the sake of simplicity when defining a property called \textit{connectivity} as follows
and writing proofs throughout the paper, without loss of generality, we can assume the existence of an initial output, $Y_0$. That is, except for the case where the state is the output, $Y_0$ is fictitious with some arbitrary distribution $Q(y_0)$ and is independent of all other variables.

Now, for both Setting I and Setting II, we assume that the NOST channels are connected.
\begin{definition}[Connectivity]
\label{definition:connceted}
A NOST channel \eqref{eq:DefNOST} is \textit{connected} if for any pair of outputs, $y',y\in \cY$, there exist an integer $T$ (shorthand for $T(y',y)\in \mathbbm{N}$) and a sequence of channel inputs $x^{T}$ (shorthand for $x^T(y',y)$) such that 
$Q_{Y_{T}|X^T,Y_0}(y|x^{T},y')>0$.
\end{definition}

Alternatively, we may formulate Definition \ref{definition:connceted} as follows. Assume, without loss of generality, that $\cY=[1:|\cY|]$, and denote by $Q$ the $|\cY|\times |\cY|$ matrix $[Q_{i j}], i,j\in \cY$, where 
\begin{align}
Q_{i j} \triangleq \max_{x \in \cX} Q_{Y_k|X_k,Y_{k-1}}(j|x,i). \label{eq:Qij}
\end{align}
A NOST channel \eqref{eq:DefNOST} is \textit{connected} if
for all $i,j \in \cY$ there exists an integer $T(i,j)$ such that $Q$ satisfies 
\begin{align}
(\overbrace{Q\cdots Q}^{T(i,j)\ \text{times} })_{i j}>0.
\end{align}
This definition is equivalent to Definition \ref{definition:connceted} since both of them imply that for any initial output $y'$ (or $i$)
and any desired output $y$ (or $j$), there exists a sequence of channel inputs such that there is a positive probability of reaching $y$ from $y'$.
\section{Main Results}
\label{sec:main_results}
In this section we present our main results pertaining to Setting I and Setting II.

\subsection{Setting I -- No CSI}
\label{subsection:MainResSet1}
The following theorem characterizes the capacity of Setting I, $C_{\text{FB}}$, as a single-letter expression.

\begin{theorem}[Feedback Capacity of Setting I]
\label{theorem:cap}
The feedback capacity of a connected NOST channel without CSI is given by
\begin{align}
C_{\text{FB}}&= \max_{P(x|y')} I(X;Y|Y'),
\label{eq:cap}
\end{align}
where the joint distribution is $P(y',x,y)=\pi(y')P(x|y') Q(y|x,y')$, $Q(y|x,y')$ is defined in \eqref{eq:y-xy}, and $\pi(y')$ is a stationary distribution induced by the Markov kernel $P(y|y')=\sum_x P(x|y') Q(y|x,y')$.
\end{theorem}
We note that $C_{\text{FB}}$ is not affected by the initial state, $s_0$. Furthermore, $Q(y|x,y')$ \eqref{eq:y-xy} given as a part of the joint distribution, implies that $S'-Y'-X$ forms a Markov chain.
By definition, $\pi(y')$ is a stationary distribution if it 
is a solution of $\pi P=\pi$, where $\pi$ is a probability vector on $\cY$, and $P$ is the probability transition matrix $P(y|y')$, whose rows and columns represent the previous and next outputs $y',y\in \cY$, respectively; thus $P_{Y'}(y')=\pi(y')=P_Y(y'), \forall y' \in \cY$.
From the assumption that $\cY$ is a finite set, there is always at least one stationary output distribution (see, e.g., \cite[Chapter~5.5]{durrett2019probability}) given any input distribution. If the stationary distribution is not unique, there are infinitely many stationary output distributions\footnote{For instance, if $|\cY|=2$ and $P(x|y')$ induces a transition matrix $P(y|y')$ given by 
the identity matrix of size $2$, $I_2$, all output distributions are stationary, i.e.,  any $\pi \triangleq [p \; \bar{p}], p\in [0,1]$ is a distribution solving $\pi I_2 = \pi$.}.
However, the following lemma states that the maximum in \eqref{eq:cap} can always be attained by an input distribution that induces a unique stationary output distribution.
\begin{lemma}
\label{lem:Eli}
For a connected NOST channel without CSI, 
\begin{align}
\max_{P(x|y')} I(X;Y|Y')=\max_{P(x|y')\in \cP_{\pi}} I(X;Y|Y'),
\end{align}
where $\cP_{\pi}$ is defined as
$\cP_{\pi}\triangleq\{P(x|y'): \text{there exists a unique stationary output distribution }\pi(y')\}$, and this set is non-empty.
\end{lemma}

A special case of Setting I is where the channel state is the output, i.e., $s_i=y_i$. 
Conversely, by the joint distribution given in Theorem \ref{theorem:cap}, it can be seen that $C_{\text{FB}}$ depends on the averaged channel $Q(y|x,y')$ \eqref{eq:y-xy} with a fictitious state $y'$; hence, Setting I and this special case are operationally equivalent. 
The setting where $s_i=y_i$ has been studied in \cite{chen2005capacity}, and the feedback capacity was derived under some assumptions on the channel; nevertheless, our result generalizes upon their seminal results in two ways. 
First, we relax the assumptions in \cite{chen2005capacity} to a connectivity condition (Definition \ref{definition:connceted}), which allows us to determine the feedback capacity of a wide family of channels, e.g., the POST($\alpha$) channel, whose feedback capacity
was derived in \cite{POSTchannel}. We discuss the relaxation issue and demonstrate the connectivity condition in Section \ref{sec:specialCases}.
Second, another contribution of our work regarding Setting I is a novel convex optimization formulation of $C_{\text{FB}}$, given in the following theorem. 
From \eqref{eq:cap} it is not clear if $I(X;Y|Y')$ is a concave function of $P(x|y')$; however, Theorem \ref{theorem:convexity} clarifies that it is, in fact, a concave function of the joint distribution, $P(y',x)$.
\begin{theorem}[Convex Optimization for $C_{\text{FB}}$]
\label{theorem:convexity}
The feedback capacity of a connected NOST channel without CSI, $C_{\text{FB}}$, can be formulated as the following convex optimization problem:
\begin{subequations}
\begin{alignat}{2}
&\!\underset{P(y',x) \in \cP(\cY \times \cX)}{\mathrm{max}}        &&  
\qquad I(X;Y|Y') \label{eq:optProb1Obejective}
\\
&\mathrm{subject}\text{ }\mathrm{to}                && \sum\limits_{x,y}P_{Y',X}(\tilde{y},x)Q_{Y|X,Y'}(y|x,\tilde{y})-\sum\limits_{y',x}P_{Y',X}(y',x)Q_{Y|X,Y'}(\tilde{y}|x,y')  = 0, \; \forall \tilde{y} \in \cY. \label{eq:condStationary}
 \end{alignat}
\end{subequations}
\end{theorem}
The benefit in formulating the feedback capacity as a convex optimization problem is the ability thus afforded to compute
it via implementing known convex optimization algorithms. Notice that the stationarity of the outputs distribution is expressed by Constraints \eqref{eq:condStationary}, which are equivalent to $P_{Y'}(\tilde{y})=P_Y(\tilde{y})$. We note that a similar approach was taken in \cite{Graph-Based}.

\subsection{Main Results of Setting II -- CSI Available at the Encoder}
The following theorem characterizes the capacity of Setting II, $C_{\text{FB-CSI}}$.
\begin{theorem}[Feedback Capacity of Setting II]
\label{theorem:cap2}
The feedback capacity of a connected NOST channel when the state information is available causally at the encoder is
\begin{align}
C_{\text{FB-CSI}}&= \max_{P(u|y'),x(u,s')} I(U;Y|Y'),
\label{eq:cap2}
\end{align}
where the joint distribution is $P(y',u,y)=\pi(y')P(u|y') P_f(y|u,y')$, in which
\begin{align}
    P_f(y|u,y')= \sum_{s',x} Q(s'|y') \mathbbm{1}\{x=f(u,s')\} Q(y|x,s'),
    \label{eq:y-uy}
\end{align}
$\pi(y')$ is a stationary distribution induced by the Markov kernel $P(y|y')=\sum_{u} P(u|y') P_f(y|u,y')$, and $U$ is an
auxiliary RV with
$|\cU|\le L \triangleq \min \{ (|\cX|^{|\cS|},(|\cX|-1)|\cS||\cY|+1, (|\cY|-1)|\cY|+1 \}$.
\end{theorem}

We note that \eqref{eq:y-uy} implies that $S'-Y'-U$ forms a Markov chain.
The feedback capacity expression in \eqref{eq:cap2} is interesting, as it combines the idea of an auxiliary RV and stationary distributions. This is the first appearance in the literature of such a combination in a single-letter capacity expression. Any auxiliary input letter $u \in \cU$ represents a distinct deterministic mapping from $\cS$ to $\cX$. Such mappings are called \textit{strategies}, and were first introduced in Shannon's work \cite{Shannon58}.
It is clear from \eqref{eq:cap2} that CSI increases the feedback capacity in the general case, because choosing $f$ to be $f(u,s)=u$, where $\cU=\cX$, i.e., $x$ and $u$ are identical, gives \eqref{eq:cap}. Furthermore, although there is generally a total of $|\cX|^{|\cS|}$ strategies, $C_{\text{FB-CSI}}$ can be achieved with at most $L$ of them. Hence, the maximization on $f(u,s')$, which is $x$, in \eqref{eq:cap2} is to choose a subset of $L$ maximizing strategies from the set of all strategies (that is, $\binom{|\cX|^{|\cS|}}{L}$ ways to choose in total).
We note that in the case where the state sequence is i.i.d., $C_{\text{FB-CSI}}$ recovers the capacity derived by Shannon \cite{Shannon58} (feedback cannot increase the capacity of DMCs \cite{shannon56}) with the cardinality bound
$|\cU|\le \min \{(|\cX|-1)|\cS|+1, |\cY|\}$ (see, e.g., \cite{el2011network}). In comparison, our general cardinality bound, $L$, has a multiplication by $|\cY|$ because of the memory preserved by the previous output, but both cardinality bounds coincide in the case of i.i.d. states.

Analogically to Lemma \ref{lem:Eli} for the case without CSI, the following Lemma declares that the maximum in \eqref{eq:cap2} can particularly be attained by an input distribution and a function $f:\cU \times \cS \to \cX$ that induce a uniqueness of the stationary output distribution.

\begin{lemma}
\label{lem:Eli2}
For a connected NOST channel with CSI available at the encoder, 
\begin{align}
\max_{P(u|y'),x(u,s')} I(U;Y|Y')=\max_{P(u|y'),x(u,s')\in \cP_{\pi}} I(U;Y|Y'),
\end{align}
where $\cP_{\pi}\triangleq\{P(u|y'),x(u,s'): \text{there exists a unique stationary output distribution }\pi(y')\}$, and this set is non-empty.
\end{lemma}

The following theorem enables us to compute $C_{\text{FB-CSI}}$, since for any choice of $f(\cdot)$ the feedback capacity expression, $\max_{P(u|y')} I(U;Y|Y')$, can be formulated as a convex optimization problem similar to that of Theorem \ref{theorem:convexity}, yet with input $U$ instead of $X$. 
\begin{theorem}[Convex Optimization for $C_{\text{FB-CSI}}$]
\label{theorem:convexity2}
For any $f:\cU \times \cS \to \cX$  with $|\cU| =  L$, the 
expression for the feedback capacity of a connected NOST channel with CSI available at the encoder, $C_{\text{FB-CSI}}$, given in \eqref{eq:cap2}, can be formulated as the following convex optimization problem:
\begin{subequations}
\begin{alignat}{2}
&\!\underset{P(y',u) \in \cP(\cY \times \cU)}{\mathrm{max}}        &&  
\qquad I(U;Y|Y') \label{eq:optProb2Obejective}
\\
&\mathrm{subject}\text{ }\mathrm{to}                && \qquad \sum\limits_{u,y}P_{Y',U}(\tilde{y},u)P_f(y|u,\tilde{y})-\sum\limits_{y',u}P_{Y',U}(y',u)P_f(\tilde{y}|u,y') = 0, \quad \forall \tilde{y} \in \cY.
\label{eq:UcondStationary}
 \end{alignat}
\end{subequations}
where 
$P_{f}(y|u,y')$, given in \eqref{eq:y-uy}, is determined by $f$ and the NOST channel law.
\end{theorem}
As a consequence of Theorem \ref{theorem:convexity2}, the feedback capacity $C_{\text{FB-CSI}}$ given in \eqref{eq:cap2} can be readily computed, because the maximum over functions $f$ that map  $x=f(u,s')$ is equivalent to taking the maximum of the solutions of all $\binom{|\cX|^{|\cS|}}{L}$ convex optimization problems with $|\cU|=L$.

In Section \ref{sec:specialCases}, we provide an example of a connected NOST channel whose $C_{\text{FB}}$ and $C_{\text{FB-CSI}}$ given in the previous theorems are equal, and derive these capacity
expressions there explicitly as detailed in Theorem \ref{theorem:noisyPOSTcap}.
\section{Examples}
\label{sec:specialCases}
This section covers special cases of Setting I and Setting II, and shows how the feedback capacity expressions of each setting, $C_{\text{FB}}$ and $C_{\text{FB-CSI}}$, subsume the corresponding capacity characterizations from the literature. Further, we explain and demonstrate the connectivity property (Definition \ref{definition:connceted}) on the POST($\alpha$) channel \cite{POSTchannel}, and generalize this channel to 
one having a state that is stochastically dependent on the output, a noisy version we thus call "the noisy-POST$(\alpha,\eta)$ channel", for which $C_{\text{FB}}$ and $C_{\text{FB-CSI}}$ are equal, i.e., CSI available at the encoder does not increase its feedback capacity.
\subsection{Special Cases of Setting I -- No CSI}
\label{subsec:set1Special}
\subsubsection{The state is a deterministic function of the output}
The case where $s_i=y_i$ is trivially a special case of Setting I. As explained after Theorem \ref{theorem:cap}, Setting I can also be formulated with the channel state $y'$ and therefore, operationally, both settings are equivalent.
In \cite{chen2005capacity}, the special case  $s_i=y_i$ was studied, but Theorem \ref{theorem:cap} generalizes their result by relaxing the assumption in \cite{chen2005capacity}. In particular, \cite{chen2005capacity} shows that the capacity expression is the one in Theorem \ref{theorem:cap} but subject to \emph{strong irreducibility} and \emph{strong aperiodicity} \cite[Defs. 2,4]{chen2005capacity}\footnote{More accurately, \cite{chen2005capacity} also assumed an additional, unnecessary condition (\cite[Def.~6]{chen2005capacity}) just for simplifying the proof, as it was remarked there that it was not crucial for the feedback capacity theorem.}. Our derivations do not require any aperiodicity assumption, and the strong irreducibility is also relaxed: recall that by \eqref{eq:Qij},
our connectivity condition holds if and only if 
\begin{align}
\forall i,j\in \cY, \; \exists T(i,j): \; (\overbrace{Q\cdots Q}^{T(i,j) \text{times} })_{i j}>0.    
\end{align}
The strong irreducibility in \cite[Def.~2]{chen2005capacity} can be written similarly by changing the maximum in \eqref{eq:Qij} to a minimum. In words, strong irreducibility requires irreducibility (in the usual sense) of the output Markov process $\{Y_i| i=0,1,\dots \}$ with respect to \textit{all} input distributions, while Definition \ref{definition:connceted} only requires the \textit{existence} of an input distribution that induces a path (a positive probability) between any two channel outputs. We proceed to show the significance of this relaxation via the following Example a), then we provide Example b) of a periodic, connected NOST channel.

\paragraph{The POST($\alpha$) channel}
\label{subsec:POST}
\begin{figure}[b]
\begin{center}
\begin{psfrags}
    \psfragscanon
    \psfrag{A}[][][0.8]{$1$}
    \psfrag{B}[][][0.8]{$2$}
    \psfrag{C}[][][0.8]{$1$}
    \psfrag{D}[][][0.8]{$2$}
    \psfrag{E}[][][1]{$s_{i-1}=1$}
    \psfrag{F}[][][1]{$x_i$}
    \psfrag{G}[][][1]{$y_i$}
    \psfrag{H}[][][0.8]{$1$}
    \psfrag{I}[][][0.8]{\raisebox{-0.4cm}{$1-\alpha$\hspace{-0.1cm}}}
    \psfrag{J}[][][0.8]{$\alpha$\hspace{-0.4cm}}
    \psfrag{E1}[][][1]{$Q(s_i|y_i)$}
    \psfrag{F1}[][][1]{$y_i$}
    \psfrag{G1}[][][1]{$s_i$}
    \psfrag{H1}[][][0.8]{$1$}
    \psfrag{I1}[][][0.8]{\raisebox{-0.4cm}{$1-\eta$\hspace{-0.1cm}}}
    \psfrag{J1}[][][0.8]{$\eta$\hspace{-0.4cm}}
    \psfrag{K}[][][0.8]{$1$}
    \psfrag{L}[][][0.8]{$2$}
    \psfrag{M}[][][0.8]{$1$}
    \psfrag{N}[][][0.8]{$2$}
    \psfrag{O}[][][1]{$s_{i-1}=2$}
    \psfrag{P}[][][1]{$x_i$}
    \psfrag{Q}[][][1]{$y_i$}
    \psfrag{R}[][][0.8]{$1-\alpha$}
    \psfrag{S}[][][0.8]{\raisebox{-0.4cm}{$1$}}
    \psfrag{T}[][][0.8]{$\alpha$\hspace{0.6cm}}
\includegraphics[scale=0.28]{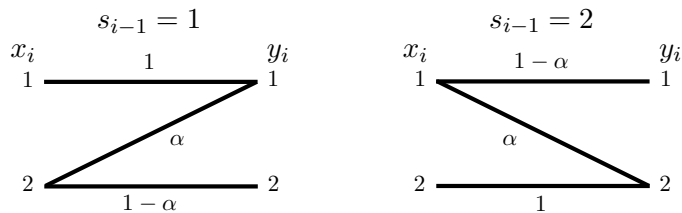}
\caption{The ZS-channel model characterizing the probability of $Q(y_i|x_i,s_{i-1})$, where $\alpha \in [0,1]$. For $s_{i-1}=1$ we have the $Z$ topology, and for $s_{i-1}=2$ we have the $S$ topology.} \label{fig:ZS-channel}
\psfragscanoff
\end{psfrags}
\end{center}
\end{figure}
The POST$(\alpha)$ channel studied in \cite{POSTchannel} is a simple, yet representative, example of an FSC. The alphabets $\cX,\cY,\cS$ are all binary, and the channel output depends on the input and the channel state as shown on Fig. \ref{fig:ZS-channel}.
Specifically, if the input and the channel state are equal, the channel output is equal to them, while otherwise it is a random instance due to parameter $\alpha \in [0,1]$.
The state evolution of the POST($\alpha$) channel is implied by its name, "Previous Output is the STate (POST)" \cite{POSTchannel}, i.e., 
$s_i=y_i$.
The POST($\alpha$) channel is not strongly irreducible under the definition of \cite[Def.~2]{chen2005capacity}, but is a connected NOST channel under Definition \ref{definition:connceted}, demonstrated as follows. For the POST($\alpha$) channel, 
Matrix $Q$ defined in \eqref{eq:Qij}, and Matrix $\tilde{Q}$ defined by replacing the maximum in \eqref{eq:Qij} with a minimum are, respectively, 
\begin{align}
Q=
\begin{bmatrix}
1 & 1-\alpha \\
1-\alpha & 1 
\end{bmatrix},    \quad
\tilde{Q}=
\begin{bmatrix}
\alpha & 0 \\
0 & \alpha 
\end{bmatrix}.    \nn
\end{align}
On the one hand, for any power of $n=1,2,\dots$, entries $(\tilde{Q}^{n})_{1 2}=(\tilde{Q}^{n})_{2 1}=0$, i.e., there is no path from output $y'=1$ to output $y=2$ (and from $y'=2$ to $y=1$); and therefore, the POST($\alpha$) is not strongly irreducible. On the other hand, $Q_{i j}>0$ for all $i,j \in \cY$ and $\alpha\in [0,1)$, and thus the POST($\alpha$) channel is connected (except for $\alpha=1$, in which case the feedback capacity is trivially $0$).
Consequently, Theorem \ref{theorem:cap} recovers its known feedback capacity, which is the closed-form capacity expression of a simple $Z$ channel, as was derived in \cite{POSTchannel}.

It is compelling that many channel instances like the trapdoor \cite{Permuter06_trapdoor_submit}, Ising \cite{Ising_channel} and and POST($\alpha$) share the same channel characterization $Q(y|x,s')$. However, their feedback capacity is fundamentally different due to the channel state evolution. 
In Section \ref{subsec:noisy}, we generalize the POST($\alpha$) channel to have a stochastic state evolution, 
and study its feedback capacity with and without CSI available at the encoder. We now proceed to Example b) of a periodic connected NOST channel that does not satisfy strong aperiodicity \cite[Def.~4]{chen2005capacity}.

\begin{figure}[b]
\begin{center}
\begin{psfrags}
    \psfragscanon
    \psfrag{A}[][][0.8]{$0$}
    \psfrag{B}[][][0.8]{$1$}
    \psfrag{C}[][][0.8]{$0$}
    \psfrag{D}[][][0.8]{$1$}
    \psfrag{E}[][][1]{$s_{i-1}=0$}
    \psfrag{F}[][][1]{$x_i$}
    \psfrag{G}[][][1]{$y_i$}
    \psfrag{H}[][][0.8]{$1-\alpha$}
    \psfrag{I}[][][0.8]{\raisebox{-0.4cm}{$1-\beta$\hspace{-0.1cm}}}
    \psfrag{I1}[][][0.8]{\raisebox{-0.4cm}{$\alpha$}\hspace{-0.3cm}}
    \psfrag{I2}[][][0.8]{\raisebox{-0.4cm}{$\beta$}\hspace{-0.2cm}}
    \psfrag{J1}[][][0.8]{$1$}
    \psfrag{J2}[][][0.8]{\raisebox{-0.3cm}{$1$}}
    \psfrag{J3}[][][0.8]{$1$}
    \psfrag{J4}[][][0.8]{\raisebox{-0.35cm}{$1$}}
    \psfrag{E1}[][][1]{$Q(s_i|y_i)$}
    \psfrag{F1}[][][1]{$y_i$}
    \psfrag{G1}[][][1]{$s_i$}
    \psfrag{K}[][][0.8]{$0$}
    \psfrag{L}[][][0.8]{$1$}
    \psfrag{M}[][][0.8]{$2$}
    \psfrag{N}[][][0.8]{$3$}
    \psfrag{O}[][][1]{$s_{i-1}=1$}
    \psfrag{P}[][][1]{$x_i$}
    \psfrag{Q}[][][1]{$y_i$}
    \psfrag{R}[][][0.8]{$1-\gamma$}
    \psfrag{S}[][][0.8]{\raisebox{-0.4cm}{$1-\delta$}}
    \psfrag{T}[][][0.8]{$\alpha$\hspace{0.6cm}}
    \psfrag{T1}[][][0.8]{$\gamma$\hspace{-0.2cm}}
    \psfrag{T2}[][][0.8]{$\delta$\hspace{-0.2cm}}
    \psfrag{U}[][][0.8]{$2$}
    \psfrag{V}[][][0.8]{$3$}
\includegraphics[scale=0.28]{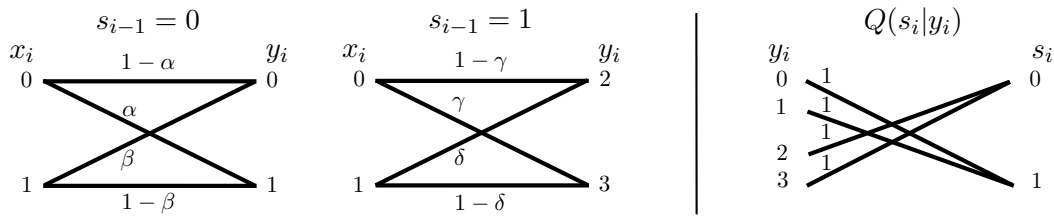}
\caption{An example of a periodic connected NOST channel. On the LHS: the conditional probabilities $Q(y_i|x_i,s_{i-1})$, i.e., for both $s_{i-1}=0,1$, general BIBO channels are obtained with some parameters $0\le \alpha,\beta,\gamma,\delta\le1$. On the RHS: the state evolution $Q(s_i|y_i)$, i.e., $s_i$ is a deterministic function of $y_i$.} \label{fig:periodic-channel}
\psfragscanoff
\end{psfrags}
\end{center}
\end{figure}
\paragraph{A periodic NOST channel}
Let $\cX=\cS=\{0,1\}, \cY=\{0:3\}$, where for both states a general binary-input binary-output channel (BIBO) is obtained, yet with different outputs, as given on the LHS of Fig. \ref{fig:periodic-channel}. 
The state $s_i$ is a deterministic function of the output $y_i$, as given on the RHS of Fig. \ref{fig:periodic-channel}, and it induces 
a periodic Markov output process with period $2$. Although the output Markov chain is periodic, Theorem \ref{theorem:cap} can determine the feedback capacity of this channel, because it is clearly a connected NOST channel. Denote the DMC capacities of the BIBOs in $s_{i-1}=0,1$ by $C_1,C_2$, respectively. Applying Theorem \ref{theorem:cap} gives that $C_{\text{FB}}$ of this periodic channel example is the average of $C_1$ and $C_2$, because any stationary output distribution $\pi(y'), y'\in \cY$ satisfies
    $\pi_{Y'}(0)+\pi_{Y'}(1)=\pi_{Y'}(2)+\pi_{Y'}(3)=0.5$;
thus the feedback capacity is
\begin{align}    
    C_{\text{FB}}^{\text{Per.}}=\sum_{y'\in \cY} \pi(y') I(X;Y|Y'=y')=\left(\pi_{Y'}(0)+\pi_{Y'}(1)\right) C_2+\left(\pi_{Y'}(2)+\pi_{Y'}(3)\right) C_1=\frac{C_1+C_2}{2}.
\end{align}

\subsubsection{The state is independent of the output}
\label{subsec:Set1IID}
In this special case, the channel state evolution satisfies $Q(s_i|y_i)=Q(s_i)$. 
Consequently, for this case, the averaged DMC $Q(y|x,y')$ in \eqref{eq:y-xy} does not depend on the previous channel input $y'$, and can be written as $Q(y|x)\triangleq\sum_{s'} Q(s')Q(y|x,s')$ (which implies that $X$ is independent of $S'$).
The term for $Q(y|x)$ averages the DMCs $Q(y|x,s')$ over the state; the capacity is
    $C=\max_{P(x)} I(X;Y)$, 
and is not increased by feedback \cite{shannon56}.
We show that $C_{\text{FB}}$ is equal to 
$C$ as follows. On the one hand, $Q(y|x,y')=Q(y|x)$ implies that $I(X;Y|Y')\le I(X;Y)$, and on the other hand we have $I(X;Y|Y')\ge I(X;Y)$, which follows by considering $P(x|y')=P(x)$ as this implies that $Y$ and $Y'$ are independent due to $P(y|y')=\sum_x P(x|y') Q(y|y',x)=\sum_x P(x) Q(y|x)=P(y)$.

\subsection{Special Cases 
of Setting II -- CSI Available at the Encoder}
We previously showed that Setting I (without CSI) was operationally equivalent to the setting where $s_i=y_i$ with feedback, 
by arguing that each one of them can be considered as a special case of the other. 
However, Setting II (CSI available at the encoder) cannot be considered a special case of the setting $s_i=y_i$ with feedback and CSI available at the encoder
due to the following explanation.
There is already a real state with a physical meaning, $s'$, known at the encoder, and we cannot introduce a new fictitious state. When CSI is not available, it follows from Theorem \ref{theorem:cap} that the probability of $Y_i$ given $(X_i,Y_{i-1})$ is determined by $Q(y|x,y')$ \eqref{eq:y-xy}, i.e., it is fixed by the NOST channel model because of the Markov chain $S'-Y'-X$, which follows since the encoder does not have access to the states. However, this Markov chain does not necessarily hold when CSI is available at the encoder, and, therefore, the probability of $Y_i$ given $(X_i,Y_{i-1})$ is given by $P(y|x,y')$ which is not fixed only by the NOST channel model, but also by the choice of an auxiliary RV $U$ that maps the real state $S'$ to a channel input $X$ by some function $f:\cU \times \cS \to \cX$, as shown in Theorem \ref{theorem:cap2}.

Another special case of Setting II is where the state is independent of the output, i.e., $Q(s_i|y_i)=Q(s_i)$.
In this case, we obtain a new DMC, $P_f(y|u,y')=\sum_{s',x} Q(s') \mathbbm{1}\{x=f(u,s')\} Q(y|x,s')=P_f(y|u)$, with input $u$ and output $y$, as can be seen from \eqref{eq:y-uy}.
This implies that $U$ and $S'$ become independent. 
The capacity in this case was derived by Shannon 
\cite{Shannon58} as $\max_{P(u),x(u,s')} I(U;Y)$, 
where $U$ is, indeed, an auxiliary RV independent of $S'$. 
Feedback does not increase the capacity of DMCs, and it can be shown that $C_{\text{FB-CSI}}$ recovers Shannon's capacity expression by using the fact that $P_f(y|u,y')=P_f(y|u)$ and
repeating the same arguments presented in Section \ref{subsec:Set1IID} with $U$ instead of $X$.
\subsection{The Noisy-POST$(\alpha,\eta)$ Channel - Special Example for which $C_{\text{FB}}$ and $C_{\text{FB-CSI}}$ are Equal}
\label{subsec:noisy}
\begin{figure}[t]
\begin{center}
\begin{psfrags}
    \psfragscanon
    \psfrag{A}[][][0.8]{$0$}
    \psfrag{B}[][][0.8]{$1$}
    \psfrag{C}[][][0.8]{$0$}
    \psfrag{D}[][][0.8]{$1$}
    \psfrag{E}[][][1]{$s_{i-1}=0$}
    \psfrag{F}[][][1]{$x_i$}
    \psfrag{G}[][][1]{$y_i$}
    \psfrag{H}[][][0.8]{$1$}
    \psfrag{I}[][][0.8]{\raisebox{-0.4cm}{$1-\alpha$\hspace{-0.1cm}}}
    \psfrag{J}[][][0.8]{$\alpha$\hspace{-0.4cm}}
    \psfrag{E1}[][][1]{$Q(s_i|y_i)$}
    \psfrag{F1}[][][1]{$y_i$}
    \psfrag{G1}[][][1]{$s_i$}
    \psfrag{H1}[][][0.8]{$1$}
    \psfrag{I1}[][][0.8]{\raisebox{-0.4cm}{$1-\eta$\hspace{-0.1cm}}}
    \psfrag{J1}[][][0.8]{$\eta$\hspace{-0.4cm}}
    \psfrag{K}[][][0.8]{$0$}
    \psfrag{L}[][][0.8]{$1$}
    \psfrag{M}[][][0.8]{$0$}
    \psfrag{N}[][][0.8]{$1$}
    \psfrag{O}[][][1]{$s_{i-1}=1$}
    \psfrag{P}[][][1]{$x_i$}
    \psfrag{Q}[][][1]{$y_i$}
    \psfrag{R}[][][0.8]{$1-\alpha$}
    \psfrag{S}[][][0.8]{\raisebox{-0.4cm}{$1$}}
    \psfrag{T}[][][0.8]{$\alpha$\hspace{0.6cm}}
\includegraphics[scale=0.28]{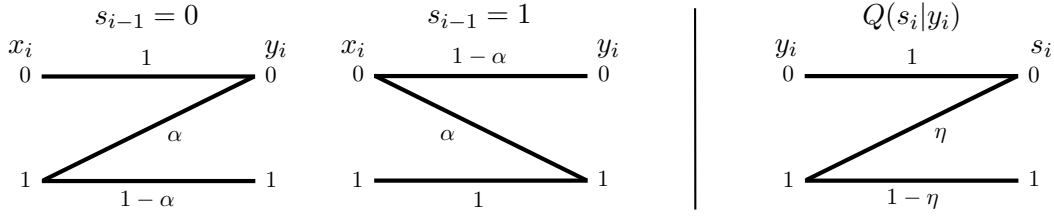}
\caption{The noisy-POST($\alpha,\eta$) channel. On the LHS: the ZS-channel model characterizing the probability of $Q(y_i|x_i,s_{i-1})$, where $\alpha \in [0,1]$. 
On the RHS: the state evolution $Q(s_i|y_i)$ as the $Z$ topology, where $\eta \in[0,1]$.} \label{fig:noisyPOST}
\psfragscanoff
\end{psfrags}
\end{center}
\end{figure}
In this section, we introduce an interesting example of a NOST channel for which $C_{\text{FB}}$ and $C_{\text{FB-CSI}}$ are equal, i.e., CSI available at the encoder does not increase its feedback capacity. This example is a generalization of the POST$(\alpha)$ channel, 
i.e., the channel output depends on the input and the channel state identically to the POST($\alpha$) channel, while the state evolution is generalized, as illustrated in Fig. \ref{fig:noisyPOST}.
We emphasize that in all previous channel instances studied in the literature, such as the trapdoor, Ising and POST($\alpha$), the state evolves according to a deterministic rule and, thus, can be determined at the encoder, while here we focus on a noisy, new version of the POST($\alpha$) channel, in which the state evolves \textit{stochastically} according to parameter $\eta \in [0,1]$. In particular, the channel state depends on the output via a Z-channel, 
i.e., if the output is zero, the next state equals the channel output, and otherwise, the next state equals the output with probability $1-\eta$. We call this generalized channel "the noisy-POST$(\alpha,\eta)$"; note that when $\eta=0$ we obtain the original "Previous Output is the STate (POST)" \cite{POSTchannel} channel. Similarly to the demonstration of the connectivity on the POST($\alpha$) in Section \ref{subsec:set1Special}, it can be verified that the noisy-POST$(\alpha,\eta)$ channel is also connected under Definition \ref{definition:connceted}.

Here, we study the feedback capacity of this noisy-POST$(\alpha,\eta)$ channel with or without CSI available at the encoder,
denoted by $C_{\text{FB-CSI}}^{\text{N-POST}}(\alpha,\eta)$ and $C_{\text{FB}}^{\text{N-POST}}(\alpha,\eta)$, respectively.
For simplicity, we arbitrarily focus on the case of $\alpha=0.5, \eta \in [0,1]$ as summarized in the following theorem, and analyze it. 
\begin{theorem}
\label{theorem:noisyPOSTcap}
For the noisy-POST$(\alpha,\eta)$ channel with any $\alpha, \eta\in [0,1]$, CSI available at the encoder does not increase the feedback capacity, and
\begin{align}
  C_{\text{FB}}^{\text{N-POST}}(0.5,\eta)&=\max_{a,b \in [0,1]} \textstyle \frac{b+\eta}{a+b+\eta} \left( H \left(\frac{a}{2}\right) -a \right)+ \frac{a}{a+b+\eta} \left( H\left( \frac{b+\eta}{2}\right)-b H\left( \frac{\overline{\eta}}{2}\right)-\overline{b} H\left(\frac{\eta}{2}\right) \right). \label{eq:noistPOSTcaphalf}
\end{align}
\end{theorem}
The first result of Theorem \ref{theorem:noisyPOSTcap}, i.e., $C_{\text{FB-CSI}}^{\text{N-POST}}(\alpha,\eta)=C_{\text{FB}}^{\text{N-POST}}(\alpha,\eta)$, 
is proved at the end of this section. The second result of Theorem \ref{theorem:noisyPOSTcap}, i.e., Eq. \eqref{eq:noistPOSTcaphalf}, follows straightforwardly from applying
Theorem \ref{theorem:cap} on the noisy-POST$(0.5,\eta)$ channel where $a\triangleq P_{X|Y'}(1|0)$ and $b \triangleq P_{X|Y'}(0|1)$ are the optimization variables; its derivation is tedious and thus is omitted.

\begin{figure}[t]
\begin{center}
    \includegraphics[scale=0.6]{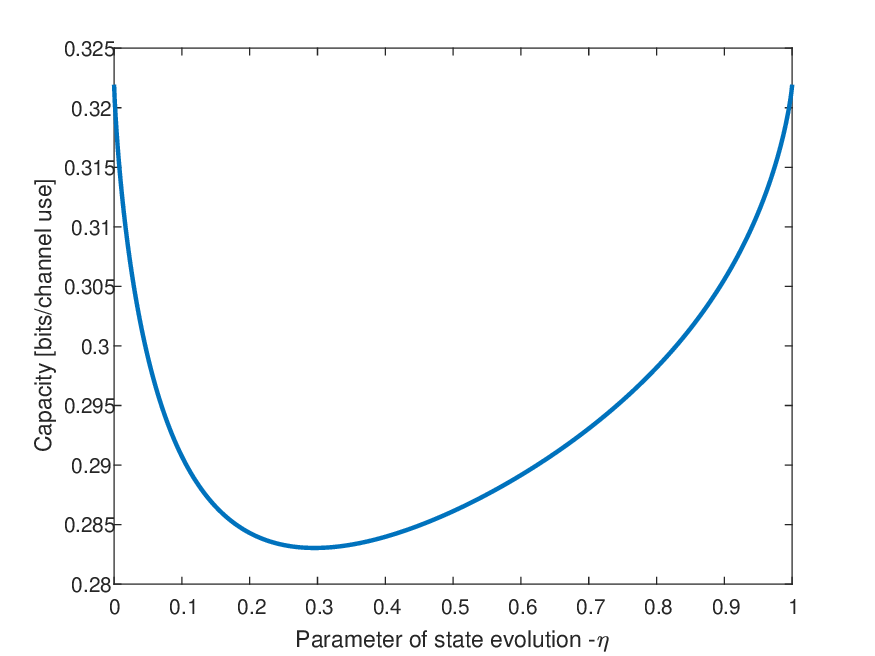}
\end{center}
\caption{The feedback capacity of the noisy-POST$(0.5,\eta)$ channel with or without CSI available at the encoder.}
\label{fig:capPOST}
\end{figure}

In Fig. \ref{fig:capPOST}, 
$C_{\text{FB}}^{\text{N-POST}}(\alpha,\eta)$ is evaluated for $\eta \in [0,1]$ using the convex optimization problem in Theorem \ref{theorem:convexity}. It can be seen that it is a convex function of $\eta$. In particular, for $\eta=0$, the POST($0.5$) channel is obtained, and for $\eta=1$ the $Z$-channel with parameter $0.5$ is obtained; in both cases, the feedback capacity is $-\log_2(0.8)\approx 0.3219$. In the case where $\eta\in(0,1)$, it can be seen that the feedback capacity is less than the capacity of the $Z$-channel. This reflects the rate-loss due to the fact that state is known at the encoder, but not at the decoder.

\begin{figure}[b]
\begin{center}
    \includegraphics[scale=0.66]{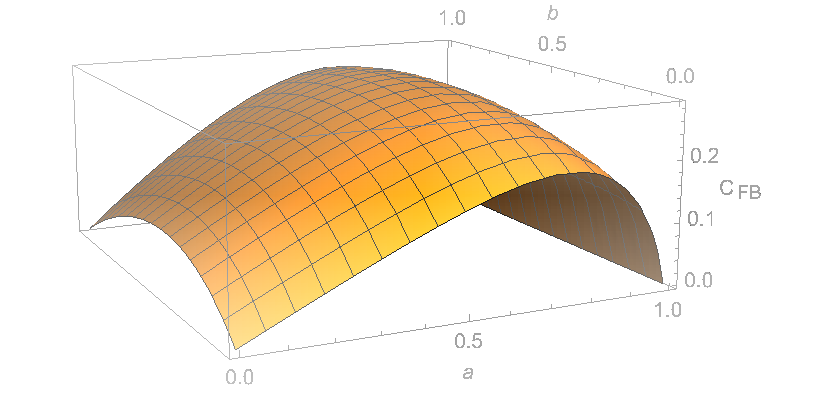}
\end{center}
\caption{The objective function of \eqref{eq:noistPOSTcaphalf} evaluated for all values of $a,b \in [0,1]$ and arbitrary $\eta=0.5$.}
\label{fig:capPOST2}
\end{figure}
For the special case $\eta=0$, the feedback capacity is achieved with $a=b=0.4$, and for $\eta=1$ it is achieves with $a=0.4, b=0.6$. For general $\eta \ne 0,1$, deriving a simpler capacity expression than \eqref{eq:noistPOSTcaphalf} is challenging. In Fig. \ref{fig:capPOST2}, we evaluate the objective function of \eqref{eq:noistPOSTcaphalf} as a function of the optimization variables $a$ and $b$, and $\eta=0.5$. It is interesting to note that although we prove the concavity of the feedback capacity in $P(y',x)$ in Theorem \ref{theorem:convexity}, Fig. \ref{fig:capPOST2} suggests that the feedback capacity of the noisy-POST$(0.5,0.5)$ is also a concave function of $P(x|y')$. A similar phenomenon is observed for other values of $\eta\in(0,1)$ as well.

\begin{table}[b]
\caption{All the strategies of binary input and binary state alphabets, $\cS=\cX=\{0,1\}$.} \centering
\label{table:strategies}
\begin{tabular}[b]{||c|c|c|c||}
\hline \hline
$x(u,s')$ & $s'=0$ & $s'=1$ \\
\hline \hline
$u_0$ & $0$ & $0$ \\
\hline
$u_1$ & $0$ & $1$ \\
\hline
$u_2$ & $1$ & $0$ \\
\hline
$u_3$ & $1$ & $1$ \\
\hline \hline
\end{tabular}
\end{table}
Next, we prove Theorem \ref{theorem:noisyPOSTcap}.
\begin{proof}[Proof of Theorem \ref{theorem:noisyPOSTcap}]
\label{appendix:ProofOfCardinalityNoisyPOST}
We prove here that CSI at the encoder does not increase the feedback capacity of the noisy-POST$(\alpha,\eta)$ channel, i.e.,  $C_{\text{FB}}^{\text{N-POST}}(\alpha,\eta)=C_{\text{FB-CSI}}^{\text{N-POST}}(\alpha,\eta)$. Consider $|\cU|=|\cX|^{|\cS|}=4$ with all possible strategies as detailed in Table \ref{table:strategies}.
By Theorem \ref{theorem:cap2}, assume that $I_{P_1}(U;Y|Y')$ is the feedback capacity of the noisy-POST$(\alpha,\eta)$ channel with CSI available at the encoder, induced by some input distribution $P_1(u|y')$ with the corresponding joint distribution
\begin{align}
    P_1(y',u,x,y)&=\pi_1(y')\sum_{s'} Q(s'|y')P_1(u|y')\mathbbm{1}\{x=f(u,s')\}Q(y|x,s'), \nn
\end{align}
where $\pi_1(y')$ is a stationary output distribution induced from the conditional output distribution $P_1(y|y')$. We construct an input distribution $P_2(x|y')$ with the corresponding conditional mutual information satisfying $I_{P_2}(X;Y|Y')=I_{P_1}(U;Y|Y')$ induced by the joint distribution
\begin{align}
    P_2(y',x,y)&=\pi_2(y')P(x|y') Q(y|x,y'), \nn
\end{align}
where $\pi_2(y')$ is a stationary output distribution induced from the conditional output distribution $P_2(y|y')$ (see Theorem \ref{theorem:cap}). 
Clearly, $I_{P_1}(U;Y|Y')\ge I_{P_2}(X;Y|Y')$; thus, our goal is to show that 
$I_{P_1}(U;Y|Y')\le I_{P_2}(X;Y|Y')$. In the construction of $P_2(x|y')$, we only demand that it satisfies
\begin{align}
    P_2(x|y')=P_1(x|y') \quad \forall x\in \cX, y' \in \cY, \label{eq:preserveP12}
\end{align}
where $P_1(x|y')$ is the input distribution induced by $P_1(u|y')$, and given by 
\begin{align}
    P_1(x|y')= \sum_{u,s'} P_1(u,s',x|y')=\sum_{u,s'} P_1(u|y')Q(s'|y')\mathbbm{1}\{x=f(u,s')\}. \nn
\end{align}
Hence, for the noisy-POST$(\alpha,\eta)$ we obtain
\begin{align}
    P_2(X=1|Y'=0)&\triangleq P_1(u_2|Y'=0)+ P_1(u_3|Y'=0), \label{eq:P2_a}\\
    P_2(X=0|Y'=1)&\triangleq P_1(u_0|Y'=1)+\eta P_1(u_1|Y'=1)+(1-\eta) P_1(u_2|Y'=1). \label{eq:P2_b}
\end{align}
From the construction in \eqref{eq:preserveP12}, it follows that the conditional output distributions are also equal, i.e.,
\begin{align}
    P_2(y|y')=\sum_{x}P_2(x|y') Q(y|x,y')=\sum_{x}P_1(x|y') Q(y|x,y')=P_1(y|y'),\; \forall y',y \in \cY. \nn 
\end{align}
Consequently, $\pi_2(y')=\pi_1(y'), \; \forall y'\in \cY$ and $H_{P_2}(Y|Y')=H_{P_1}(Y|Y')$ hold; thus,
\begin{align}
    I_{P_1}(U;Y|Y')&=H_{P_2}(Y|Y')-H_{P_1}(Y|Y',U) \nn\\
    & \stackrel{(a)} \le H_{P_2}(Y|Y')-H_{P_2 }(Y|Y',X) \nn\\
    &= I_{P_2}(X;Y|Y') , \nn 
\end{align}
where (a) follows from defining 
$q \triangleq  H_{P_1}(Y|Y',U)-H_{P_2}(Y|Y',X)\ge 0$. 
We show that $q\ge 0$ by applying the noisy-POST($\alpha,\eta$) channel model on 
\begin{align}
    H_{P_1}(Y|Y',U)=&\sum_{y'} \pi_1(y') H_{P_1}(Y|Y'=y',U), \nn\\
    H_{P_2}(Y|Y',X)=&\sum_{y'} \pi_2(y') H_{P_2}(Y|Y'=y',X)=\sum_{y'} \pi_1(y') H_{P_2}(Y|Y'=y',X), \nn\\
  P_f(y|u,y')=&\sum_{s',x} Q(s'|y') \mathbbm{1}\{x=f(u,s')\} Q(y|x,s'), \nn
\end{align}
giving the following identities:
\begin{align}
    H_{P_1}(Y|y'=0,U)
    =&P_1(X=0|Y'=0)\stackrel{(a)}= P_2(X=0|Y'=0)=H_{P_2}(Y|y'=0,X), \nn\\
    H_{P_1}(Y|y'=1,U)= & \textstyle P_1(u_0|y'=1) H(\frac{1-\eta}{2})+P_1(u_1|y'=1)H(\eta)+P_1(u_2|y'=1) +P_1(u_3|y'=1) H(\frac{\eta}{2}) \nn\\
    \stackrel{(b)}=& \textstyle P_1(u_0|y'=1) ( H(\frac{1-\eta}{2})-H(\frac{\eta}{2}) ) +P_1(u_1|y'=1) ( H(\eta)-H(\frac{\eta}{2}) ) \nn\\
    & \textstyle +P_1(u_2|y'=1) ( 1-H(\frac{\eta}{2}) ) +H(\frac{\eta}{2}), \nn\\
    H_{P_2}(Y|y'=1,U)=& \textstyle P_2(X=0|Y'=1) H(\frac{1-\eta}{2}) +P_2(X=1|Y'=1) H(\frac{\eta}{2}) \nn\\
    \stackrel{(c)}=& \textstyle P_2(X=0|Y'=1) H(\frac{1-\eta}{2}) +\left(1-P_2(X=0|Y'=1)\right) H(\frac{\eta}{2}) \nn\\
    \stackrel{(d)}=& \textstyle P_1(u_0|y'=1)(H(\frac{1-\eta}{2})-H(\frac{\eta}{2})) + P_1(u_1|y'=1)\eta (H(\frac{1-\eta}{2})-H(\frac{\eta}{2})) \nn\\
    &+ \textstyle P_1(u_2|y'=1)) (1-\eta) (H(\frac{1-\eta}{2})-H(\frac{\eta}{2}))+ H(\frac{\eta}{2}), \nn
\end{align}
where (a) and (d) follow from the construction of $P_2(x|y')$
in \eqref{eq:P2_a}-\eqref{eq:P2_b}; and (b) and (c) follow from substituting $P_1(u_3|y'=1)$ and $P_2(X=1|Y'=1)$, respectively, with their complementary distribution to $1$.
Hence, we deduce that
\begin{equation}
    q= P_1(y'=1) \left(H_{P_1}(Y|y'=1,U)-H_{P_2}(Y|y'=1,X)\right) \ge0, \nn
\end{equation}
because
\begin{align}
    &H_{P_1}(Y|y'=1,U)-H_{P_2}(Y|y'=1,X) \nn\\
    & =\textstyle P_1(u_1|y'=1) ( H(\eta)-(1 - \eta) H(\frac{\eta}{2}) -\eta H(\frac{1-\eta}{2})) + P_1(u_2|y'=1) ( 1-\eta H(\frac{\eta}{2}) -(1-\eta) H(\frac{1-\eta}{2})) \nn\\
    & \ge \textstyle P_1(u_1|y'=1) ( H(\eta)-(1 - \eta) H(\frac{\eta}{2}) -\eta H(\frac{1-\eta}{2})) \nn\\
    &\stackrel{(a)} \ge \textstyle P_1(u_1|y'=1) ( H(\eta)- H(\eta(1-\eta)) \nn\\
    &\stackrel{(b)}
    \ge 0, \nn 
\end{align}
where (a) follows from the concavity of the binary entropy, and (b) is due to $H(\eta)\ge H(\eta(1-\eta))$, which is trivial for $\eta \in [0,0.5]$, and for $\eta \in [0.5,1]$ it is also trivial after using $H(\eta)=H(1-\eta)$.

To conclude, $I_{P_2}(X;Y|Y')=I_{P_1}(U;Y|Y')$, which implies that CSI available at the encoder does not increase the feedback capacity of the noisy-POST($\alpha,\eta$) channel.
\end{proof}

\section{Proofs}
\label{sec:proof_cap}
In this section, we prove our main results given in Section \ref{sec:main_results}. 
In particular, the proofs of the feedback capacity expression, i.e., Theorems \ref{theorem:cap} and \ref{theorem:cap2}, are given in 
Sections \ref{subsec:capStandAlone} and 
\ref{subsec:capCSI}, respectively.
We note that Lemmas \ref{lem:Eli} and \ref{lem:Eli2} are used to establish the achievability proofs of the mentioned Theorems \ref{theorem:cap} and \ref{theorem:cap2}, respectively. As Lemma \ref{lem:Eli2} generalizes Lemma \ref{lem:Eli}, we only prove the former in Section \ref{subsec:eli}. The proof of the cardinality bound of Theorem \ref{theorem:cap2} is provided in Section \ref{subsection:CardinalityBound}.
Finally, Section \ref{subsection:convexity} proves the convex optimization formulations of the feedback capacity expressions, i.e., Theorems \ref{theorem:convexity} and \ref{theorem:convexity2}.
Before all these proofs are given, we introduce the following useful lemma, whose proof is given in Appendix \ref{appendix:MarkovChannel}.
\begin{lemma}
\label{lemma:Helpful}
    For any NOST channel \eqref{eq:DefNOST} in Setting I (without CSI), 
\begin{align}
     Q(y_i|x^{i},y^{i-1},m)&=
     \sum_{s_{i-1} \in \cS} Q(s_{i-1}|y_{i-1})Q(y_i|x_i,s_{i-1})
     \nn\\
     &=Q(y_i|x_i,y_{i-1}). \label{eq:MarkovChannel}
\end{align}
\end{lemma}

\subsection{Proof of Theorem \ref{theorem:cap}}
\label{subsec:capStandAlone}
\subsubsection{Proof of Converse}
Throughout the proof, the initial output, $y_0$, is assumed to be available at both the encoder and the decoder.

For a fixed sequence of $(2^{nR},n)$ codes, where $R$ is an achievable rate, we bound $R$ as 
\begin{align}
    R -\epsilon_{n} &\stackrel{(a)}\le \frac{1}{n} I(M;Y^n) \nn\\
    &\stackrel{(b)}=\frac{1}{n} \sum_{i=1}^n I(M,X_i;Y_i|Y^{i-1}) \nn\\
    &\stackrel{(c)}\le \frac{1}{n} \sum_{i=1}^n H(Y_i|Y_{i-1})-H(Y_i|Y_{i-1},X_i) \nn\\
    &\le \max_{\{P(x_i|y_{i-1})\}_{i=1}^{n}} \frac{1}{n} \sum_{i=1}^n I(X_i;Y_i|Y_{i-1})  ,\label{eq:converseAlternative1}
\end{align}
where $\epsilon_n$ tends to zero as $n \to \infty$, and
\begin{enumerate}[label={(\alph*)}]
\item follows from Fano's inequality;
\item follows from the fact that $x_i$ is a deterministic function of $(m,y^{i-1})$; and
\item follows from the fact that conditioning reduces entropy and from Lemma \ref{lemma:Helpful}.
\end{enumerate}
In Section \ref{subsec:capCSI}, we show a fundamental result on the optimality of time-invariant input distributions in Lemma \ref{lemma:ORON}. To avoid repetition, we refer the reader to follow the proof of Lemma \ref{lemma:ORON} with $x$ instead of $u$, the joint distribution $P(y',x,y)=P(y',x) \sum_{s'} Q(s'|y') Q(y|x,s')$ and
the modified set 
\begin{align}
\mathcal{D}_{\epsilon} \triangleq \{ P(y',x) \in \mathcal{P}_{\cY \times \cX} : |P_{Y'}(y)- \sum_{y',x} P_{Y',X}(y',x) \sum_{s'} Q(s'|y') Q(y|x,s') |\le \epsilon, \forall y \}, \nn
\end{align}
in order to deduce that any achievable rate $R$ must satisfy $R \le \max_{P(x|y')} I(X;Y|Y')$.
\qed

\subsubsection{Proof of Achievability}
We need to prove that rates satisfying
$R<\max_{P(x|y')} I(X;Y|Y')$ are achievable.
However, recalling Lemma \ref{lem:Eli}, which states that it is sufficient to maximize over input distribution that induce a unique stationary output distribution, we prove, for simplicity, that rates satisfying
$R<\max_{P(x|y')\in \cP_{\pi}} I(X;Y|Y')$ are achievable. 
The proof uses rate-splitting where $Y^{n}$ is treated as a time-sharing sequence, and the previous channel output, $Y_{i-1}=y'\in \cY$, determines one of $|\cY|$ DMCs that are multiplexed at the encoder and demultiplexed at the decoder. 
\begin{proof}
At time $i=1$, the encoder transmits an arbitrary input symbol $X_1$, and afterwards $Y_1$ is known both at the decoder and at the encoder by the feedback. More generally, from time $i=2$ on, the previous channel output $Y_{i-1}$ is known at both parties before each transmission. According to the known $Y_{i-1}=y'$, a DMC characterized by $Q_{Y|X}=Q_{Y|X,Y'=y'}$ with message $M_{y'}$ is treated in the current channel use.

\textit{Rate-splitting and code construction:} 
Fix an input distribution $P(x|y')\in \cP_{\pi}$ that achieves $C_{\text{FB}}$ in \eqref{eq:cap}, i.e., a collection of conditional PMFs $P(x|y')$ on $\cX$ for
every $y' \in \cY$ is to be determined such that a unique 
stationary distribution on the outputs, $\pi(y')$, is induced.
By Lemma \ref{lem:Eli}, such $P(x|y')$ always exists on account of the connectivity assumption (Definition \ref{definition:connceted}) and the assumption that $|\cY|$ is finite.
Each message $M$ consists of $|\cY|$ independent sub-messages $M_{y'}\in [1:2^{nR_{y'}}]$, $y' \in \cY$. This implies that $R=\sum_{y'} R_{y'}$.
From the achievability of the channel coding theorem for DMCs, 
in each DMC $Q_{Y|X,Y'=y'},y'\in \cY$, every rate $R_{y'}<I(X;Y|Y'=y')$ is achievable, where the joint distribution is determined by the fixed conditional input distribution given $y'$, i.e., $P(x,y|y')=P(x|y')Q(y|x,y')$. That is, there exists a sequence of $(2^{nR_{y'}},n)$ codes with an average probability of error $P(\hat{M}_{y'}\ne {M}_{y'})$ that tends to zero as $n\to \infty$. For a block length $n$ and $y'\in \cY$, denote the codebook of the $n$th code of such a sequence by $\mathcal{C}_{n,y'}$. Each $\mathcal{C}_{n,y'}$ consists of $2^{nR_{y'}}$ codewords $x^n(m_{y'})$.

Returning to our connected NOST channel: to send message $m=\{m_{y'}| y'\in \cY\}$, at time $i\in [2,n+1]$, with known previous output $Y_{i-1}=y'$, the encoder transmits the next unsent symbol
of codeword $x^n(m_{y'})\in \mathcal{C}_{n,y'}$. Upon receiving the entire output sequence $y^{n+1}$, the decoder demultiplexes it into $|\cY|$ sub-sequences of outputs $y^{n_{y'}}(y')$ whose previous output is $y'\in \cY$, where $n_{y'}$ is the number of times that output $y'$ was "visited" during times $i\in [1:n]$, i.e.,
\begin{align}
    n_{y'}=\sum_{i=2}^{n+1} \mathbbm{1}{\{Y_{i-1}=y'\}}, 
\end{align}
thus $\sum_{y'} {n_{y'}}=n$.
For each sub-sequence $y^{n_{y'}}(y'), y'\in \cY$ of DMC $Q_{Y|X,Y=y'}$, the receiver decodes $m_{y'}$ as in the aforementioned direct coding theorem for DMCs (joint typicality decoding).

\textit{Analysis of the probability of error:} Using this theorem, it follows that the probability of error in decoding each $m_{y'}$ tends to zero as $n \to \infty$ if \begin{align}
   R_{y'} &\le \lim_{n \to \infty} \frac{n_{y'}}{n}I(X;Y|Y'=y')-\delta(\epsilon), \nn\\
   &\stackrel{(a)}=\pi(y') I(X;Y|Y'=y')-\delta(\epsilon),
   \label{eq:DMC_Rate}
\end{align}
where $\delta(\epsilon)$ tends to zero as $\epsilon \to 0$.
Step (a) follows from Birkhoff’s ergodic theorem on Markov chains with a unique stationary distribution (see, e.g., \cite{durrett2019probability}), since the fixed $P(x|y')$ induces a homogeneous Markov chain $\{Y_i| i=0,1,\dots \}$ with a unique stationary output distribution, $\pi(y')$.
Now, the total probability of error in decoding the message $m=\{m_{y'}| y'\in \cY\}$ tends to zero as $n \to \infty$ if
\begin{align}
    R = \sum_{y'} R_{y'} \le \sum_{y'}
    \pi(y') I(X;Y|Y'=y')-\tilde{\delta}(\epsilon)=  I(X;Y|Y')-\tilde{\delta}(\epsilon),
\end{align}
where $\tilde{\delta}(\epsilon)$ tends to zero as $\epsilon \to 0$. This completes the proof of achievability.
\end{proof}

We note that Theorem \ref{theorem:cap} can be derived by another approach based on the directed information, which generally characterizes the capacity of channels with feedback, as given in Appendix \ref{appendix:DI_Proof}.

\subsection{Proof of Theorem \ref{theorem:cap2}}
\label{subsec:capCSI}
\subsubsection{Proof of Converse} 
Here, we prove that 
an achievable rate $R$ must satisfy 
$R \le  \max_{P(u|y')}  I(U;Y|Y')$, where, without loss of generality, $\cU$ is the 
set of all strategies. In this proof, the initial output, $y_0$, is assumed to be available at both the encoder and the decoder, and the proof consists of two parts. In the first part, we show that for achievable rates 
\begin{equation}
\label{eq:UBregular}
R\leq \max_{{\{P(u_i|y_{i-1})\}}_{i=1}^n}
 \frac{1}{n}  \sum_{i=1}^n I(U_i;Y_i|Y_{i-1}) +\epsilon_n,   
\end{equation}
where $U_i\in \cU$ enumerates all possible 
strategies and maps $S_{i-1}$ to $X_i$,
and $\epsilon_n\to 0$ as $n \to \infty$. The second part of the proof, stated by the following Lemma \ref{lemma:ORON}, is to show that it is sufficient to maximize over time-invariant conditional distributions, $P(u|y')$.
\begin{lemma}
For a connected NOST channel with CSI available at the encoder,
\label{lemma:ORON}
\begin{align}
\label{eq:lemmaOron}
    \lim_{n \to \infty} \max_{{\{P(u_i|y_{i-1})\}}_{i=1}^n}
 \frac{1}{n}  \sum_{i=1}^n I(U_i;Y_i|Y_{i-1})
    &\le \max_{P(u|y')}  I(U;Y|Y'),
\end{align}
where $U_i\in \cU$ enumerates all possible mappings from $\cS$ to $\cX$, the joint distribution on the RHS is $P(y',u,y)=\pi(y')P(u|y') P_f(y|u,y')$, and $P_f(y|u,y')$ is given in \eqref{eq:y-uy}.
\end{lemma}
The proof of Lemma \ref{lemma:ORON} is based on a method developed in \cite{OronStationaryConverse}, and it is given next in the second part of Section \ref{subsec:capCSI}. The first part of the converse, i.e., Inequality \eqref{eq:UBregular}, is now shown.

For a fixed sequence of $(2^{nR},n)$ codes such that the probability of error $P_e^{(n)}\to 0$ as $n \to \infty$,
we bound
\begin{align}
    R-\epsilon_n
    &\stackrel{(a)}\le\frac{1}{n} \sum_{i=1}^n H(Y_i|Y^{i-1})-H(Y_i|Y^{i-1},M) \nonumber\\
    &\le \frac{1}{n}  \sum_{i=1}^n H(Y_i|Y_{i-1})-H(Y_i|Y^{i-1},M) \nonumber\\
    &\stackrel{(b)}= \frac{1}{n}  \sum_{i=1}^n I(U_i;Y_i|Y_{i-1}) \nonumber\\
    &\stackrel{(c)}\leq 
    \max_{{\{P(u_i|y_{i-1}),P(x_i|u_i,s_{i-1})\}}_{i=1}^n} \frac{1}{n}  \sum_{i=1}^n I(U_i;Y_i|Y_{i-1}) \nonumber\\
       &\stackrel{(d)}= 
    \max_{{\{P(u_i|y_{i-1}),P(v_i),x_i=f_i(u_i,v_i,s_{i-1})
        \}}_{i=1}^n} \frac{1}{n}  \sum_{i=1}^n I(U_i;Y_i|Y_{i-1}) \nonumber\\
       &\stackrel{(e)}\leq 
    \max_{{\{P(\Tilde{u}_i|y_{i-1}),x_i=f_i(\Tilde{u}_i,s_{i-1})
        \}}_{i=1}^n} \frac{1}{n}  \sum_{i=1}^n I(\Tilde{U}_i;Y_i|Y_{i-1}) \nonumber\\
     &\stackrel{(f)}= \max_{{\{P(\Tilde{u}_i|y_{i-1}),x_i=f(s_{i-1},\Tilde{u}_i,i)\}}_{i=1}^n} \frac{1}{n}  \sum_{i=1}^n I(\tilde{U}_i;Y_i|Y_{i-1})  \nonumber\\
 &\stackrel{(g)}\leq \max_{{\{P(\dbtilde{u}_i|y_{i-1}),x_i=f(\dbtilde{u}_i,s_{i-1})\}}_{i=1}^n}
 \frac{1}{n}  \sum_{i=1}^n I(\dbtilde{U}_i;Y_i|Y_{i-1}) \label{UB:last_step}
\end{align}
where $\epsilon_n\to 0$ as $n \to \infty$, and
\begin{enumerate}[label={(\alph*)}]
\item follows from Fano's inequality;
\item follows from defining $U_i\triangleq(M,Y^{i-1})$ for every $i \in [1:n]$; this definition also satisfies the Markov chain $(Y_i,S_i)-(X_i,S_{i-1})-U_i$ due to the assumption that the channel is an FSC; 
\item follows from 
the following lemma, whose proof is given in the next part of Section \ref{subsec:capCSI}: 
\begin{lemma}
\label{lem:UpperBound}
For any $k$, the joint distribution $P(y_{k-1},u_k,y_k)$ is determined by
\[
\{P(u_i|y_{i-1})P(x_i|u_i,s_{i-1})\}_{i=1}^k;
\]
\end{lemma}
\item follows from the Functional Representation Lemma \cite{el2011network}, i.e., for every $i \in [1:n]$ there exists a RV $V_i$, such that $X_i$ can be represented as a function of $(U_i,S_{i-1},V_i)$, where $V_i$ is independent of $(U_i,S_{i-1})$, and the Markov chain $(Y_i,S_i)-(U_i,S_{i-1},X_i)-V_i$ holds (hence $(Y_i,S_i)-(X_i,S_{i-1})-(V_i,U_i)$ holds as well),
and from the following lemma, whose proof 
uses the aforementioned properties of $V_i$ and is similar to that of Lemma \ref{lem:UpperBound}, and therefore it is omitted: 
\begin{lemma}
\label{lem:UpperBound2} 
For any $k$, the joint distribution $P(y_{k-1},u_k,y_k)$ is determined by 
\[
\{P(u_i|y_{i-1})P(v_i)x_i(v_i,u_i,s_{i-1})\}_{i=1}^k;
\]
\end{lemma}
\item follows from defining $\Tilde{U_i}\triangleq(U_i,V_i)$; hence, 
\[P(\Tilde{u}_i|y_{i-1})=P(v_i|y_{i-1})P(u_i|v_i,y_{i-1}),\]
as $P(u_i|y_{i-1})$ and $P(v_i)$ are sub-domains of $P(u_i|v_i,y_{i-1})$ and $P(v_i|y_{i-1})$, respectively; 
\item follows since there exists a time-invariant function $f$ such that $f(\Tilde{u},s,i)=f_i(\Tilde{u}_i,s_{i-1})$; and
\item follows from defining $\dbtilde{U}=(\Tilde{U}_i,T=i)$, where $T$ represents the time index.
\end{enumerate}
For simplicity of appearance, we replace $\dbtilde{U}_i$ with $U_i$ and obtain from \eqref{UB:last_step} that any achievable rate must satisfy \eqref{eq:UBregular}, where $\cU$ is the aforementioned set of all strategies, i.e., $|\cU|=|\cX|^{|\cS|}$ (increasing the cardinality of $\cU$ beyond $|\cX|^{|\cS|}$ cannot increase the objective function further).
Finally, the proof is completed by Lemma \ref{lemma:ORON}. \qed
\subsubsection{Proofs of Technical Lemmas \ref{lemma:ORON}-\ref{lem:UpperBound}}
\begin{proof}[Proof of Lemma \ref{lemma:ORON}]
The proof is divided into two parts. In the first part, $\frac{1}{n}  \sum_{i=1}^n I(U_i;Y_i|Y_{i-1})$ is upper bounded for any $n$ and joint distribution on $(U^n,Y^n)$. Subsequently, in the second part of the proof we take the limit of this bound when $n$ tends to infinity in order to obtain \eqref{eq:lemmaOron}. 

The first part of the proof is as follows. For any $n$ and ${\{P(y_{i-1},u_i)\}}_{i=1}^n$,
\begin{align}
\label{eq:step}
    \frac{1}{n}  \sum_{i=1}^n I(U_i;Y_i|Y_{i-1}) &\stackrel{(a)}\le I(U;Y|Y') \nn\\
    &\stackrel{(b)}\le \max_{P \in \mathcal{D}_{\frac{1}{n}}}  I(U;Y|Y'),
\end{align}
where for Steps
\begin{enumerate}[label={(\alph*)}]
\item the joint distribution on the RHS is $\tilde{P}(y',u,y)=\tilde{P}(y',u) \sum_{s',x} Q(s'|y') \mathbbm{1}\{x=f(u,s')\} Q(y|x,s')$, in which
\begin{align}
\tilde{P}(y',u)\triangleq \frac{1}{n} \sum_{i=1}^n P_{Y_{i-1},U_i}(y',u), \label{eq:averageDist}  
\end{align}
$\cU$ is defined identically to all $\cU_i$, i.e., it is the set of all strategies mapping $s'$ to $x$ by the deterministic function $f$, and this step follows from the fact that
$I(U;Y|Y')$ is concave in the joint distribution $P(y',u)$ as is explained in particular in the proof of Theorem \ref{theorem:convexity2} given in Section \ref{subsection:convexity};
\item the notation $\mathcal{D}_{\epsilon}$ denotes the set
\begin{align}
\mathcal{D}_{\epsilon} \triangleq &\{ P(y',u) \in \mathcal{P}_{\cY \times \cU} :\left|P_{Y'}(y)- \sum_{y',u} P_{Y',U}(y',u) Q(y|u,y') \right|\le \epsilon, \forall y \}, \nn
\end{align}
where $Q(y|u,y')\triangleq \sum_{s',x} Q(s'|y') \mathbbm{1}\{x=f(u,s')\} Q(y|x,s')$, and for any codebook of length $n$, its induced probability, $\tilde{P}(y',u)$, lies in $\mathcal{D}_{\frac{1}{n}}$, i.e., 
$|\tilde{P}_{Y'}(y)- \sum_{y',u} \tilde{P}_{Y',U}(y',u) Q(y|u,y')|\le \frac{1}{n}
$ for all $y$, because by using the definition of $\tilde{P}(y',u)$ given in \eqref{eq:averageDist} we obtain
\begin{align}
\left|\tilde{P}_{Y'}(y)-\sum_{y',u} \tilde{P}_{Y',U}(y',u) Q(y|u,y') \right|
    =& \frac{1}{n} \left| \sum_{y',u} \sum_{i=1}^{n}  P_{Y_{i-1}}(y)- P_{Y_{i-1},U_i}(y',u)Q(y|u,y') \right| \nn \\
    =& \frac{1}{n} \left| \sum_{i=1}^{n} \sum_{y',u} P_{Y_{i-1}}(y)- P_{Y_{i-1},U_i}(y',u)Q(y|u,y') \right| \nn \\
    =& \frac{1}{n} \Bigg| \sum_{i=2}^{n} \sum_{y',u} P_{Y_{i-1}}(y)- P_{Y_{i-2},U_{i-1}}(y',u) Q(y|u,y') \nn\\
    & +P_{Y_0}(y)- P_{Y_{n-1},U_n}(y',u) Q(y|u,y')
    \Bigg| \nn \\
    =& \frac{1}{n} \left| P_{Y_0}(y)-P_{Y_{n}}(y)\right| \nn \\
    \le& \frac{1}{n},
\end{align}
which follows from $\sum_{y',u} P_{Y_{i-1},U_{i}}(y',u) Q(y|u,y')=P_{Y_i}(y)$ for any $i \in [1:n]$. 
\end{enumerate}
This completes the first part of the proof. In the second part of the proof, we obtain from \eqref{eq:step} 
\begin{align}
    \lim_{n \to \infty} \max_{{\{P(u_i|y_{i-1})\}}_{i=1}^n}
 \frac{1}{n}  \sum_{i=1}^n I(U_i;Y_i|Y_{i-1})
    &\le\lim_{n\to \infty}
    \max_{P \in \mathcal{D}_{\frac{1}{n}}} I(U;Y|Y') \nn\\
    &\stackrel{(a)}= \max_{P \in \mathcal{D}_{0}} I(U;Y|Y') \nn\\
    &\stackrel{(b)} = \max_{P(u|y')}  I(U;Y|Y'),
\end{align}
where 
\begin{enumerate}[label={(\alph*)}]
\item follows since any $P\in \mathcal{D}_0$ satisfies $P\in \cap_{n=1}^{\infty} \mathcal{D}_{\frac{1}{n}}$ due to the fact that $\frac{1}{n}$ is positive for all $n$,
and vice versa, i.e., any $P\in \cap_{n=1}^{\infty} \mathcal{D}_{\frac{1}{n}}$ satisfies $P\in \mathcal{D}_0$
because $\frac{1}{n}$ monotonically decreases in $n$; hence,
$\lim_{n\to \infty} \mathcal{D}_{\frac{1}{n}} =  \mathcal{D}_{0}$;
\item follows since $\mathcal{D}_{0}$ implies the set of all $P(y',u)$ that have a stationary output distribution, i.e., $P_{Y'}(y')=P_{Y}(y')$; recall that since the output set $\cY$ is assumed to be finite, there always exists a stationary output distribution (not necessarily unique) 
with regard to any $P(u|y')$, thus $\mathcal{D}_{0}$ is non-empty.
\end{enumerate}
This concludes the proof.
\end{proof}

\begin{proof}[Proof of Lemma \ref{lem:UpperBound}]
We prove by induction that the joint distribution
$P(y_{k-1},u_k,y_k)$ is determined by $\{P(u_i|y_{i-1})P(x_i|u_i,s_{i-1})\}_{i=1}^k$, where $U_i\triangleq(M,Y^{i-1},y_0)$ and 
$y_0$ is assumed to be known at both the encoder and the decoder.
For $k=1$,
\begin{align}
    P(u_1,y_1|y_0)&=\sum_{s_0,x_1} P(s_0,u_1,x_1,y_1|y_0) \nn\\
    &=\sum_{s_0,x_1} Q(s_0) P(u_1|y_0) P(x_1|u_1,s_0) Q(y_1|x_1,s_0), \nonumber 
\end{align}
which follows from the facts that: 
$U_1=(M,y_0)$, where $M$ and $S_0$ are independent and $Y_1-(X_1,S_0)-M$ forms a Markov chain due to the FSC Markov property \eqref{eq:DefFSC}. 
Suppose that the lemma is true for $k-1$, i.e., $P(y_{k-2},u_{k-1},y_{k-1})$ is determined by $\{P(u_i|y_{i-1}),P(x_i|u_i,s_{i-1})\}_{i=1}^{k-1}$. Then, for $k$ we have
\begin{align}
P(y_{k-1},u_k,y_k) &= \sum_{s_{k-1},x_k}
P(y_{k-1},s_{k-1},u_k,x_k,y_k) \nn\\
&= \sum_{s_{k-1},x_k} P(y_{k-1})Q(s_{k-1}|y_{k-1})P(u_k|y_{k-1})P(x_k|u_k,s_{k-1}) \label{eq:jointInduction1}
Q(y_k|x_k,s_{k-1}),
\end{align}
which follows from the NOST channel Markov property \eqref{eq:DefNOST} and the definition of $U_k$.
From the induction hypothesis, 
$P(y_{k-1})$ is determined by $\{P(u_i|y_{i-1}),P(x_i|u_i,s_{i-1})\}_{i=1}^{k-1}$. Hence, from \eqref{eq:jointInduction1} it can be seen that $P(y_{k-1},u_k,y_k)$ is determined by $\{P(u_i|y_{i-1}),P(x_i|u_i,s_{i-1})\}_{i=1}^{k}$, which completes the proof.
\end{proof}

\subsubsection{Proof of Achievability}
\begin{figure}[t]
\begin{center}
\begin{psfrags}
    \psfragscanon
    \psfrag{E}[][][0.7]{$M$}
    \psfrag{S}[][][0.7]{\begin{tabular}{@{}l@{}}
    $S^{i-1}$\hspace{0.1cm}
    \end{tabular}}
    \psfrag{A}[\hspace{2cm}][][0.75]{Encoder}
	 \psfrag{F}[\hspace{2cm}][][0.75]{$U_i(M,Y^{i-1})$}
	 \psfrag{W}[\hspace{2cm}][][0.75]{$Q(y_i|x_i,s_{i-1})Q(s_i|y_i)$}
	 \psfrag{M}[\hspace{2cm}][][0.75]{NOST Channel}
	 \psfrag{P}[\hspace{2cm}][][0.75]{NOST Channel $Q(y_i|u_i,s_{i-1})Q(s_i|y_i)$\hspace{-0.25cm}}
	 \psfrag{G}[][][0.75]{$Y_i$}
	 \psfrag{C}[\hspace{2cm}][][0.75]{Decoder}
	 \psfrag{R}[\hspace{2cm}][][0.75]{$S_{i-1}$}
	 \psfrag{K}[][][0.75]{$\hat{M}$}
	 \psfrag{H}[\hspace{2cm}][][0.75]{$Y_i$}
	 \psfrag{D}[\hspace{2cm}][][0.75]{Delay}
	 \psfrag{Z}[\hspace{2cm}][][0.75]{$X_i$}
	 \psfrag{J}[\hspace{2cm}][][0.75]{$Y_{i-1}$}
	 \psfrag{L}[\hspace{2cm}][][0.75]{NOST Channel}
	 \psfrag{T}[\hspace{2cm}][][0.75]
	 {$x_i=f(u_i,s_{i-1})$}
\includegraphics[scale=0.7]{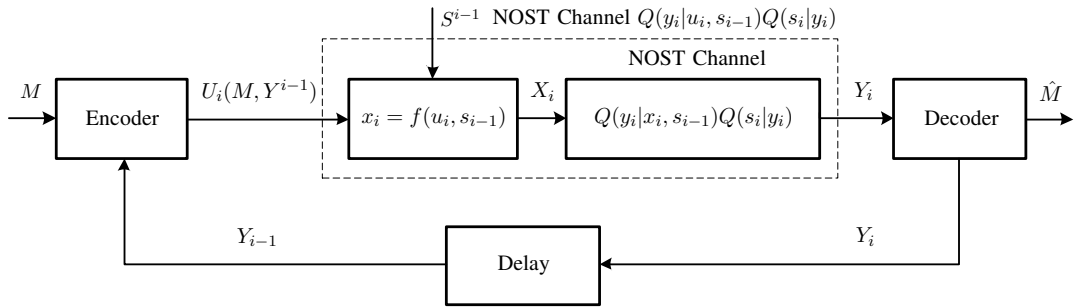}
\caption{An equivalent setting for the lower bound formulation with a new NOST channel $Q(y|u,s')Q(s|y)$ where CSI is unavailable.} \label{fig:shannon_strategy}
\psfragscanoff
\end{psfrags}
\end{center}
\end{figure}
We prove that every rate $R<\max_{P(u|y')} I(U;Y|Y')$, where $\cU$ is the set of all strategies, is achievable. This is shown by converting Setting II into a setting of type I, i.e., where no CSI is available, as is shown in Fig. \ref{fig:shannon_strategy}. 
In particular, 
at time $i\in [2:n]$ (the communication setting starts at time $i=2$ for the same reason given in the proof of achievability of Setting I), the channel input is $U_i$, which is a function of the message and feedback only, without the state, i.e., $U_i(M,Y^{i-1})$.
Then, given the current state $S_{i-1}$, the strategy $U_i$ maps $S_{i-1}$ to input $X_i$, thus inducing a new NOST channel, $Q(y|u,s')Q(s|y)$, with input $U_i$ (rather than $X_i$), in the presence of feedback. The new NOST channel is also connected, because $\cU$ specifically includes all $|\cX|$ strategies that map all states to an input $x\in \cX$. This allows us to use the achievability of Theorem \ref{theorem:cap} (in which Lemma \ref{lem:Eli2} should be used instead of Lemma \ref{lem:Eli}) and deduce that rates that satisfy $R<\max_{P(u|y')} I(U;Y|Y')$ are achievable.
\qed

\subsection{Proofs of Lemma \ref{lem:Eli} and Lemma \ref{lem:Eli2}}
\label{subsec:eli}
Here, we prove Lemma \ref{lem:Eli2} which also generalizes Lemma \ref{lem:Eli}.

\begin{proof}[Proof of Lemma \ref{lem:Eli2}]
Without loss of generality, we assume that 
$\cU$ is the set of all strategies, thus 
all strategies are chosen and 
$|\cU|=|\cX|^{|\cS|}$.
We need to prove that
\begin{align}
\label{eq:statU}
\max_{P(u|y')} I(U;Y|Y')=\max_{P(u|y')\in \cP_{\pi}} I(U;Y|Y'),
\end{align}
where $\cP_{\pi}$ denotes 
the non-empty set of input distributions $P(u|y')$ that
induce a unique stationary output distribution.
We prove it by constructing some $P(u|y')\in \cP_{\pi}$ that achieves the maximum on the LHS of \eqref{eq:statU}.

Since $\cY$ is assumed to be a finite set, 
any input distribution induces at least one stationary output distribution (see, e.g., \cite[Chapter~5.5]{durrett2019probability}).
Let $P^*(u|y')$ be an optimal input distribution that achieves the maximum on the LHS of \eqref{eq:statU}, 
denoted by $I^*(U;Y|
Y')$, and induces at least one stationary output distribution, i.e., $P^*_{Y'}(y')=P^*_Y(y'), \forall y' \in \cY$. If $P^*(u|y')$ induces
a probability transition matrix $P^*(y|y')$ whose stationary output distribution is unique, the proof is concluded. Hence, we assume, otherwise, that $P^*(y|y')$ has infinitely many stationary output distributions.  We show that there always exists $\dbtilde{P}(u|y')\in \cP_{\pi}$, i.e., an input distribution that induces a 
unique stationary output distribution $\dbtilde{\pi}(y')=\dbtilde{P}_{Y'}(y')=\dbtilde{P}_Y(y'), \forall y' \in \cY$ with the corresponding conditional mutual information 
$\dbtilde{I}(U;Y|Y')$, such that 
$\dbtilde{I}(U;Y|Y')=I^*(U;Y|Y')$.

From Markov theory
(see, e.g., \cite[Theorems 5.3.3,5.3.12]{durrett2019probability}),
since $\cY$ is finite, if $P^*(y|y')$ induces only one irreducible subset of $\cY$, there is a unique stationary output distribution on $\cY$; this contradicts our assumption. Therefore, $P^*(y|y')$
induces at least two disjoint irreducible closed subsets of $\cY$.
Assume that $P^*(u|y')$ induces two
irreducible closed subsets of $\cY$, denoted by $\cC_i, i\in \{1,2\}$,
with the corresponding probability transition matrices $P_{\cC_i}^*(y|y'), \forall y',y\in \cC_i$ which are derived from $P^*(y|y')$. Let $\pi_{\cC_i}^*(y'), \forall y'\in \cC_i$, where $\sum_{y' \in \cC_i} \pi_{\cC_i}^*(y')=1$, be the unique stationary distribution induced by $P_{\cC_i}^*(y|y')$, and denote the corresponding
maximal conditional mutual information of each $\cC_i$ by
$I^*_{\mathcal{C}_i}(U;Y|Y')\triangleq \sum_{y' \in C_i} \pi_{\cC_i}^*(y') I^*(U;Y|Y'=y')$.
It follows that $I^*(U;Y|Y')=\max_{i\in\{1,2\}}I^*_{\cC_i}(U;Y|Y')$ since 
if, without loss of generality, $I^*_{\cC_1}(U;Y|Y')\le I^*_{\cC_2}(U;Y|Y')$, then
\begin{equation}
P^*_{Y'}(y')= \begin{cases}
\pi_{\cC_2}^*(y'), & y'\in \cC_2\\
0, & \text{otherwise}
\end{cases}
\label{eq:P_Yc2}
\end{equation}
is a legitimate stationary distribution, i.e., it satisfies 
$P^*_{Y'}(y')=P^*_Y(y'), \forall y' \in \cY$.
Furthermore, we can construct $\dbtilde{P}(u|y')\in \cP_{\pi}$ such that $\dbtilde{I}(U;Y|Y')=I^*_{\cC_2}(U;Y|Y')$ as follows.
Construct $\dbtilde{P}(u|y')$ exactly as $P^*(u|y')$ for all $y' \in \cY$, but with $\dbtilde{P}(u|y'_1), y'_1\in \cC_1$ that induces a positive probability to reach an arbitrary $y_2 \in \cC_2$ from an arbitrary initial $y'_1\in \cC_1$ with some input sequence. This construction is legitimate because the NOST channel is assumed to be connected, and the $|\cX|$ strategies that map all states to a specific input $x$ are also a part of the set of all strategies. 
This construction of $\dbtilde{P}(u|y')$ renders all outputs in $\cC_1$ transient and $\cC_2$ a unique irreducible closed subset in $\cY$. That is, $\dbtilde{P}(u|y')$ induces a unique stationary output distribution on $\cY$ as given in \eqref{eq:P_Yc2}, and $\dbtilde{I}(U;Y|Y'=y')=I^*_{\cC_2}(U;Y|Y'=y')$ for any $y'\in \cC_2$. Hence, $\dbtilde{I}(U;Y|Y')=I^*(U;Y|Y')$ as desired.

The construction can be extended in the case of multiple disjoint irreducible closed subsets of $\cY$ as there can be at most $|\cY|$ subsets, which is a finite number. Hence, it can be deduced that \eqref{eq:statU} holds. Finally, using the cardinality bound of Theorem \ref{theorem:cap2}, whose proof
(given next in this section) shows that $I^*(U;Y|Y')$ can be achieved with at most $L\le |\cX|^{|\cS|}$ strategies such that $P^*(y|y')$ is preserved, we conclude that $I^*(U;Y|Y')$ can always be achieved with some $P(u|y'), x(u,s') \in \cP_{\pi}$ where $|\cU|\le L$, which concludes the proof.
\end{proof} 

\subsection{Cardinality Bound}
\label{subsection:CardinalityBound}
As the cardinality bound $|\cU|\le |\cX|^{|\cS|}$ is trivial (there are $|\cX|^{|\cS|}$ strategies in total), here we prove the
non-trivial cardinality bounds on $\cU$ in Theorem \ref{theorem:cap2},
%
i.e., $|\cU|\le   (|\cX|-1)|\cS||\cY|+1$ and $|\cU|\le (|\cY|-1)|\cY|+1$. If either of them is less than $|\cX|^{|\cS|}$, it implicitly means that not all strategies are required in order to achieve the feedback capacity, but only $L$ of them at most.
\begin{proof}[Proof of Cardinality Bounds]
We invoke the support lemma \cite[p. 631]{el2011network},
which is a consequence of the Fenchel-Eggle-ston-Caratheodory theorem \cite{eggleston}, twice, for the auxiliary RV $U$. In each use, we show how the measures of the feedback capacity, i.e., the conditional entropies in $I(U;Y|Y')=H(Y|Y')-H(Y|Y',U)$, are preserved, thereby implying both non-trivial cardinality bounds. In other words, assuming $U$ takes values in an arbitrary alphabet $\cU$, we prove that 
given any $(Y',S',U,X)$, there exists 
$(Y',S',\tilde{U},X)$ with $|\tilde{\mathcal{U}}|\le \min \{(|\cX|-1)|\cS||\cY|+1, (|\cY|-1)|\cY|+1 \}$ such that $I(U;Y|Y')=I(\tilde{U};Y|Y')$.

We begin with proving $|\cU|\le  (|\cX|-1)|\cS||\cY|+1$. $\tilde{U}$ must have $(|\cX|-1)|\cS||\cY|$ letters to preserve $P(x|s',y')$ for all $s',y'\in \cS \times \cY$.
If $P(x|s',y')$ is preserved, $P(y',y)$ is preserved as well, because  $P(y',y)=\pi(y')P(y|y')$, where $\pi(y')$ is a stationary distribution induced from the probability transition matrix $P(y|y')$ that can also be expressed by
\begin{align}
    P(y|y')&=\sum_{s',x} Q(s'|y') P(x|s',y') Q(y|x,s'). \nn
\end{align}
Additionally, if $P(y',y)$ is preserved, $H(Y|Y')$ is preserved, too. Finally, $\tilde{U}$ must have another letter to preserve $H(Y|Y',U)$. This concludes the 
the proof of the first non-trivial cardinality bound.

Following so, we prove $|\cU|\le (|\cY|-1)|\cY|+1$. However, this time, we aim to preserve $P(y|y')$ for all $y',y\in \cY$, directly. To address this, $\tilde{U}$ must have $(|\cY|-1)|\cY|$ letters, thereby preserving $H(Y|Y')$. Finally, $\tilde{U}$ must have another letter to preserve $H(Y|Y',U)$, which concludes the proof. \end{proof}

\subsection{
Convex Optimization Formulations}
\label{subsection:convexity}
Firstly, we prove Theorem \ref{theorem:convexity}, and subsequently Theorem \ref{theorem:convexity2} follows likewise.
\begin{proof}[Proof of Theorem \ref{theorem:convexity}]
In this proof we show that \eqref{eq:optProb1Obejective}-\eqref{eq:condStationary} is a convex optimization problem, i.e., the maximization domain is convex, Constraints \eqref{eq:condStationary} are linear
and the objective \eqref{eq:optProb1Obejective} is concave. The maximization is over the probability simplex ${\cP_{\cY \times \cX}}$, which is convex, and Constraints \eqref{eq:condStationary} are linear in ${\cP_{\cY \times \cX}}$ because
$Q(y|x,y')$ is the averaged channel in \eqref{eq:y-xy}. Finally, we show that the objective function in \eqref{eq:optProb1Obejective} is concave as follows.
The conditional mutual information characterizing the feedback capacity can be written as $I(X;Y|Y')=H(Y|Y')-H(Y|Y',X)$. 
The first conditional entropy can be expressed as 
$H(Y|Y')=\log |\cY|-D(P(y',y)||P(y')U(y))$,
where $U(y)=\frac{1}{|\cY|}$ is the uniform distribution over $\cY$, because
\begin{align}
  D(P(y',y)||P(y')U(y)) &=\sum_{y',y} P(y',y) \log \frac{P(y|y')}{U(y)}=\log|\cY|-H(Y|Y'). \nn
\end{align}
(This conditional entropy identity is an extension of the known entropy identity 
$D(P(y)||U(y))=\log |\cY|-H(Y)$, see \cite[Eq. (2.93)]{CovThom06}).
The relative entropy above is convex in the pair $(P(y',y), P(y')U(y))$, which are linear in $P(y',x)$ due to 
$P(y',y)=\sum_x P(y',x) Q(y|x,y')$ and $P(y')=\sum_x P(y',x)$, 
and thus $H(Y|Y')$ is concave in $P(y',x)$. On the other hand, the second conditional entropy is 
\begin{align}
    H(Y|Y',X)&=\sum_{y',x} P(y',x) H_Q(Y|x,y'), 
\end{align}
where $H_Q(Y|x,y')$ is a constant determined by $Q(y|x,y')$; i.e., $H(Y|Y',X)$ is linear in $P(y',x)$. Thus, 
the difference between both conditional entropies is concave in $P(y',x)$, which completes the proof.
\end{proof}
The proof of Theorem  \ref{theorem:convexity2} is similar to that of Theorem \ref{theorem:convexity}, but with replacing all occurrences of $x$ with $u$ and referring to $P_f(y|u,y')$ \eqref{eq:y-uy}, which is determined by the NOST channel model for a fixed $f$, instead of referring to the averaged channel $Q(y|x,y')$.
\section{Conclusions and Further Work}
\label{sec:conclusions}
A family of FSCs called \emph{NOST channels} is introduced, and their feedback capacities is derived as single-letter formulas under a connectivity condition in two scenarios: with and without CSI availability at the encoder. These formulas are shown to be computable by formulating the capacity expressions as convex optimization. Furthermore, it is demonstrated via the noisy-POST channel that CSI at the encoder may not increase the feedback capacity.

This work is part of an ongoing progress on the feedback capacity of FSCs. Obtaining a general computable capacity formula with or without CSI is challenging, and even the capacity of particular channels cannot always be obtained in a closed form. We remark here on several interesting research directions that follow from the current work. 
It may be possible to extend our capacity results to a countable channel output alphabet $\cY$. This is not straightforward from the current derivation, as we explain.
We previously mentioned, from Markov theory, that for a finite Markov chain there always exists at least one stationary distribution. However, for an infinite Markov chain, there may be no stationary distributions at all\footnote{A simple such known example is the symmetric random walk on the integers: $P_{Y|Y'}(i-1|i)=P_{Y|Y'}(i+1|i)=0.5, \forall i\in \cY\triangleq \mathbb{Z}$, which has no solution $\pi$ of $\pi P=\pi$ that is a legitimate distribution. We note that in this example, $\cY$ is irreducible. While for a finite Markov chain, irreducibility induces the existence of a unique stationary distribution, a fact we use throughout our derivations, this is not necessarily true for an infinite set. 
In fact, an irreducible Markov chain has a stationary distribution if and only if it is positive recurrent (see, e.g., \cite[Th.~5.5.12]{durrett2019probability}).}.
Another interesting research direction is to study NOST channels without feedback or in the regimes of noisy and delayed feedback links.
\appendices
\section{Proof of Lemma \ref{lemma:Helpful}}
\label{appendix:MarkovChannel}
\begin{align}
    Q(y_i|x^{i},y^{i-1},m) &=\sum_{s_{i-1} \in S} Q(s_{i-1}|x^{i},y^{i-1},m)Q(y_i|x^{i},y^{i-1},s_{i-1}) \nn \\
    &\stackrel{(a)}=\sum_{s_{i-1} \in S} Q(s_{i-1}|y^{i-1},m)Q(y_i|x^{i},y^{i-1},s_{i-1}) \nn \\
&\stackrel{(b)}=\sum_{s_{i-1} \in S} Q(s_{i-1}|y_{i-1})Q(y_i|x_i,s_{i-1}) \nn\\
&\stackrel{(c)}=Q(y_i|x_{i},y_{i-1}),
\end{align}
where
\begin{enumerate}[label={(\alph*)}]
\item follows from \eqref{eq:enc}, i.e., for each time $i$, $x_i$ is a function of $(m,y^{i-1})$.
\item follows from the NOST channel model \eqref{eq:DefNOST},
\item follows from the fact that $(x^{i-1},y^{i-2},m)$ do not appear in the summation.
\end{enumerate}
\qed
\section{Appendix - Proof of Theorem \ref{theorem:cap} Based on the Directed Information}
\label{appendix:DI_Proof}
The \textit{directed information} from $X$ to $Y$ conditioned on $S$, introduced by Massey \cite{massey1990causality} and employed with conditioning in \cite{PermuterWeissmanGoldsmith09}, is defined as 
\begin{equation}
I(X^n\rightarrow Y^n|S)\triangleq \sum_{i=1}^{n} I(X^i;Y_i|Y^{i-1},S).
\end{equation}
The \textit{causally conditional distribution}, introduced in \cite{Kramer03,permuter2006capacity}, is defined as
\begin{equation}
    P(x^n||y^{n-1})\triangleq \prod_{i=1}^n P(x_i|x^{i-1},y^{i-1}). 
\end{equation}
The feedback capacity of any FSC was shown in \cite{PermuterWeissmanGoldsmith09} to be bounded by
\begin{align}
\label{eq:capFSC}
    \lim_{n \to \infty} \frac{1}{n}  {\max_{\substack{P(x^n||y^{n-1})}}} \min_{s_0} I(X^n\rightarrow Y^n|s_0) \le C_{\text{FB}} \le \lim_{n \to \infty} \frac{1}{n}  {\max_{\substack{P(x^n||y^{n-1})}}} \max_{s_0} I(X^n\rightarrow Y^n|s_0).
\end{align}
For Setting I with any initial state $s_0$, we obtain 
\begin{align}
    \label{eq:DI-equations}
    \lim_{n \to \infty} \frac{1}{n} \max_{\substack{P(x^n||y^{n-1})}}     I(X^n\rightarrow Y^n|s_0)&= \lim_{n \to \infty} \frac{1}{n} \max_{\substack{P(x^n||y^{n-1})}} \sum_{i=1}^n I(X^i;Y_i|Y^{i-1},s_0) \nn\\
    &\stackrel{(a)}=\lim_{n \to \infty} \frac{1}{n} \max_{\substack{P(x^n||y^{n-1})}} \sum_{i=1}^n H(Y_i|Y^{i-1},s_0)-H(Y_i|Y_{i-1},X_i) \nn\\
    &\stackrel{(b)}= \lim_{n \to \infty} \max_{\{P(x_i|y_{i-1})\}_{i=1}^{n}} \frac{1}{n} \sum_{i=1}^n I(X_i;Y_i|Y_{i-1})  \nn\\
    &\stackrel{(c)}= \max_{\substack{P(x|y')}} I(X;Y|Y'),
\end{align}
where,
\begin{enumerate}[label={(\alph*)}]
\item follows from Lemma \ref{lem:Eli};
\item is explained by justifying the direct ($\ge$) and the converse ($\le$): the direct follows from maximizing over $\{P(x_i|y_{i-1}) \}_{i=1}^n$ for all $i$, which is a sub-domain of $P(x^n||y^{n-1})$, hence, $H(Y_i|Y^{i-1},s_0)=H(Y_i|Y_{i-1})$ due to the Markov chain 
\begin{align}
P(y_i|y^{i-1},s_0)&=\sum_{s_{i-1},x_i} Q(s_{i-1}|y_{i-1})P(x_i|y_{i-1})Q(y_i|x_i,s_{i-1}) =P(y_i|y_{i-1});    \label{eq:MarkovY}
\end{align}
while the converse follows from $H(Y_i|Y^{i-1},s_0)\le H(Y_i|Y_{i-1})$, and then identifying that for all $i$, the summand $I(X_i;Y_i|Y_{i-1})$ is induced by $P(y_{i-1},x_i,y_i)=P(y_{i-1})P(x_i|y_{i-1})Q(y_i|x_i,y_{i-1
})$; and
\item is also explained by justifying the direct and the converse:
the direct follows from maximizing over time-invariant input distributions that induce a unique stationary output distribution, then applying Lemma \ref{lem:Eli}; while the converse follows from maximizing over all time-invariant input distributions and applying Lemma \ref{lemma:ORON}.
\end{enumerate}
Finally, it is concluded 
from \eqref{eq:capFSC} and \eqref{eq:DI-equations}
that $\max_{\substack{P(x|y')}} I(X;Y|Y')$ is the feedback capacity, which is independent of 
the initial state, $s_0$, because \eqref{eq:DI-equations} holds for any $s_0$ and it does not affect the objective. \qed

\newpage

\bibliographystyle{IEEEtran}
\bibliography{NOST_submit}

\end{document}